\documentclass[11pt]{article}

\usepackage{amsfonts,amsmath,amssymb}
\usepackage{times}
\usepackage{natbib} 
\usepackage[latin1]{inputenc}
\usepackage[hylinks,full,titlepage,nofullpage,nousetoc,amsthm]{tcs}
\usepackage{hyperref}

\usepackage{color}

\usepackage[letterpaper,
	hmargin=0.95in,vmargin=1.0in]{geometry}

\providecommand{\remove}[1]{}

\providecommand{\ie}{\emph{i.e.}{} }
\providecommand{\eg}{\emph{e.g.}{} }

\providecommand{\supp}{\mathrm{supp}}
\providecommand{\Exp}{\mathbb{E}}

\providecommand{\dans}{\rightarrow}

\providecommand{\eps}{\varepsilon}
\providecommand{\zo}{\{0, 1\}}

\providecommand{\drawn}[1]{\stackrel{#1}{\leftarrow}}
\providecommand{\drawnr}{\drawn{R}}
\renewcommand{\Pr}{\mathop{\mathrm{Pr}}}

\def\FullBox{\hbox{\vrule width 8pt height 8pt depth 0pt}}

\def\qed{\ifmmode\qquad\FullBox\else{\unskip\nobreak\hfil
\penalty50\hskip1em\null\nobreak\hfil\FullBox
\parfillskip=0pt\finalhyphendemerits=0\endgraf}\fi}

\makeatletter
\@addtoreset{equation}{section}
\makeatother

\newcommand{\E}{{\mathbf E}}
\newcommand{\pr}{\mathbf{Pr}}

\newcommand{\pub}{\mathsf{pub}}
\newcommand{\err}{\mathsf{err}}

\newcommand{\Rand}{\mathrm{R}}
\newcommand{\Dist}{\mathrm{D}}
\newcommand{\CC}{\mathrm{CC}}

\newcommand{\calF}{\mathcal{F}}
\newcommand{\calC}{C}
\newcommand{\calD}{\mathcal{D}}
\newcommand{\D}{\mathcal{D}}
\newcommand{\calH}{\mathcal{H}}

\newcommand{\calh}{\mathfrak{h}}

\newcommand{\DRDim}{\mathrm{DRDim}}
\newcommand{\PRDim}{\mathrm{PRDim}}
\newcommand{\VCDim}{\mathrm{VC}}
\newcommand{\SCDP}{\mathrm{SCDP}}
\newcommand{\LDim}{\mathrm{LDim}}
\newcommand{\Ldim}{\mathrm{LDim}}

\newcommand{\Eval}{\mathsf{Eval}}

\DeclareMathOperator*{\argmax}{argmax}
\newcommand{\HS}{\mathsf{HS}}
\newcommand{\BALL}{\mathsf{BALL}}
\newcommand{\BOX}{\mathsf{BOX}}

\newcommand{\Thr}{\mathsf{Thr}}
\newcommand{\Point}{\mathsf{Point}}
\newcommand{\Ent}{\mathrm{H}}
\newcommand{\GT}{\mathsf{GT}}
\newcommand{\AI}{\mathsf{AugIndex}}
\newcommand{\Info}{\mathbf{I}}

\newcommand{\R}{{\mathbb R}}
\newcommand{\Z}{{\mathbb Z}}
\newcommand{\Line}{\mathsf{Line}}

\newcommand{\Div}{\mathrm{D}}

\newcommand{\bfi}{\mathbf{i}}
\newcommand{\bfa}{\mathbf{a}}
\newcommand{\bfb}{\mathbf{b}}
\newcommand{\bfx}{\mathbf{x}}
\newcommand{\bfy}{\mathbf{y}}
\newcommand{\bfz}{\mathbf{z}}
\newcommand{\bfm}{\mathbf{m}}

\newcommand{\al}{\alpha}
\newcommand{\be}{\beta}

\title{Sample Complexity Bounds on Differentially Private Learning via Communication
  Complexity}

\author{Vitaly Feldman \thanks{IBM Research - Almaden. E-mail:
    \texttt{vitaly@post.harvard.edu}. Part of this work done
    while visiting LIAFA, Universit\'{e} Paris 7.}
  \and David Xiao\thanks{CNRS, Universit\'{e} Paris 7. E-mail:
    \texttt{dxiao@liafa.univ-paris-diderot.fr}. Part of this work done
    while visiting Harvard's Center for Research on Computation and
    Society (CRCS).} }

\begin{document}

\maketitle

\begin{abstract}
In this work we analyze the sample complexity of classification by differentially private algorithms. Differential privacy is a strong and well-studied notion of privacy introduced by \cite{DworkMNS:06} that ensures that the output of an algorithm leaks little information about the data point provided by any of the participating individuals. Sample complexity of private PAC and agnostic learning was studied in a number of prior works starting with \citep{KLNRS:11}. However, a number of basic questions still remain open \citep{BeimelKN:10,ChaudhuriHsu:11,BeimelNS:13,BeimelNS:13approx}, most notably whether learning with privacy requires more samples than learning without privacy.

We show that the sample complexity of learning with (pure) differential privacy can be arbitrarily higher than the sample complexity of learning without the privacy constraint or the sample complexity of learning with approximate differential privacy. Our second contribution and the main tool is an equivalence between the sample complexity of (pure) differentially private learning of a concept class $C$ (or $\SCDP(C)$) and the randomized one-way communication complexity of the evaluation problem for concepts from $C$. Using this equivalence we prove the following bounds:
\begin{itemize}
\item $\SCDP(C) = \Omega(\LDim(C))$, where $\LDim(C)$ is the Littlestone's dimension characterizing the number of mistakes in the online-mistake-bound learning model \citep{Littlestone:87}. Known bounds on $\LDim(C)$ then imply that $\SCDP(C)$ can be much higher than the VC-dimension of $C$.
\item For any $t$, there exists a class $C$ such that $\LDim(C)=2$ but $\SCDP(C) \geq t$.
\item For any $t$, there exists a class $C$ such that the sample complexity of (pure) $\alpha$-differentially private PAC learning is $\Omega(t/\alpha)$ but the sample complexity of the approximate $(\alpha,\beta)$-differentially private PAC learning is $O(\log(1/\beta)/\alpha)$. This resolves an open problem from \citep{BeimelNS:13approx}.
\end{itemize}
\remove{
We also obtain simpler proofs for a number of known related results. Our equivalence builds on a characterization of sample complexity by \citet{BeimelNS:13} and our bounds rely on a number of known results from communication complexity.
}
\footnotetext{Preliminary version of this work has appeared in Conference on Learning Theory (COLT), 2014}
\end{abstract}


\section{Introduction}
\label{sec:intro}
In machine learning tasks, the training data often consists of information collected
from individuals. This data can be highly sensitive, for example in the case
of medical or financial information, and therefore privacy-preserving data analysis is becoming
an increasingly important area of study in machine learning, data mining and statistics \citep{DworkSmith:09,SarwateC:13,DworkRoth:14}.

In this work we focus on the task of learning to classify from labeled examples. Two standard and closely related models of this task are PAC learning \citep{Valiant:84} and agnostic \citep{Haussler:92,KearnsSS:94} learning.  In the PAC learning model the algorithm is given random examples in which each point is sampled i.i.d.~from some unknown distribution over the domain and is labeled by an unknown function from a set of functions $C$ (called concept class). In the agnostic learning model the algorithm is given examples sampled i.i.d.~from an arbitrary (and unknown) distribution over labeled points. The goal of the learning algorithm in both models is to output a hypothesis whose prediction error on the distribution from which examples are sampled is not higher (up to an additive $\eps$) than the prediction error of the best function in $C$ (which is $0$ in the PAC model). See \sectionref{sec:learn-model} for formal definitions.

We rely on the well-studied differential privacy model of privacy. Differential privacy gives a formal semantic guarantee of privacy, saying intuitively that no single individual's data has too large of an effect on the output of the algorithm, and therefore observing the output of the algorithm does not leak much
information about an individual's private data \citep{DworkMNS:06} (see \sectionref{sec:prelims-dp} for the formal
definition). The downside of this desirable guarantee is that for some problems achieving it has an additional cost: both in terms of the number of examples, or sample complexity, and computation.

The cost of differential privacy in PAC and agnostic learning was first studied by \cite{KLNRS:11}. They showed that the sample complexity\footnote{For now we ignore the dependence on other parameters and consider them to be small constants.} of differentially privately learning a concept class $C$ over domain $X$, denoted by $\SCDP(C)$, is $O(\log(|C|))$ and left open the natural question of whether $\SCDP(C)$ is different from the VC dimension of $C$ which, famously, characterizes the sample complexity of learning $C$ (without privacy constraints).  By Sauer's lemma, $\log(|C|) = O(\VCDim(C) \cdot \log(|X|))$ and therefore the multiplicative gap between these two measures can be as large as $\log(|X|)$.

Subsequently, \cite{BeimelKN:10} showed that there exists a large concept class, specifically single points, for which the sample complexity of learning with privacy is a constant. They also show that differentially private {\em proper} learning (the output hypothesis has to be from $C$) of single points $\Point_b$ and threshold functions $\Thr_b$ on the set $I_b=\{0,1,\ldots,2^b-1\}$ requires $\Omega(b)$ samples. These results demonstrate that the sample complexity can be lower than $O(\log(|C|))$ and also that lower bounds on the sample complexity of proper learning do not necessarily apply to non-proper learning that we consider here. A similar lower bound on proper learning of thresholds on an interval was given by \cite{ChaudhuriHsu:11} in a continuous setting where the sample complexity becomes infinite. They also showed that the sample complexity can be reduced to essentially $\VCDim(C)$ by either adding distributional assumptions or by requiring only the privacy of the labels.

The upper bound of \cite{BeimelKN:10} is based on an observation from \citep{KLNRS:11} that if there exists a class of functions $H$ such that for every $f \in C$ and every distribution $\calD$ over the domain, there exists $h \in H$ such that $\Pr_{x\sim \calD}[f(x) \neq h(x)] \leq \eps$ then the sample complexity of  differentially private PAC learning with error $2\eps$ can be reduced to $O(\log(|H|)/\eps)$. They refer to such $H$ as an $\eps$-representation of $C$, and define the (deterministic) $\eps$-representation dimension of $C$, denoted as $\DRDim_{\eps}(C)$, as $\log(|H|)$ for the smallest $H$ that $\eps$-represents $C$. We note that this natural notion can be seen as a distribution-independent version of the standard notion of $\eps$-covering of $C$ in which the distribution over the domain is fixed \cite[\eg][]{BenedekI:91}.

\cite{BeimelNS:13} then defined a probabilistic relaxation of $\eps$-representation defined as follows. A distribution $\calH$ over sets of boolean functions on $X$ is said to $(\eps,\delta)$-probabilistically
represent $C$ if for every $f \in C$ and distribution $\calD$ over $X$, with probability $1-\delta$ over the choice of
$H \drawnr \calH$, there exists $h \in H$ such that $\Pr_{x \sim \calD}[h(x) \neq f(x)] \leq \eps$. The $(\eps,\delta)$-probabilistic representation dimension $\PRDim_{\eps,\delta}(C)$ is the minimal $\max_{H \in  \supp(\calH)} \log |H|$, where the minimum is over all $\calH$ that $(\eps,\delta)$-probabilistically represent $C$.
\remove{Rather surprisingly\footnote{
While many other sample complexity bounds in learning theory rely on covering numbers their lower bound does not involve the standard step of constructing a large packing implied by covering. It is unclear to us if a packing implies a covering of the same size in this distribution-independent setting (as it does in the case of metric covering).}
} \cite{BeimelNS:13} demonstrated that $\PRDim_{\eps,\delta}(C)$ characterizes the sample complexity of differentially private PAC learning. In addition, they show that $\DRDim$ can upper-bounded by $\PRDim$ as $\DRDim(C) =O(\PRDim(C) + \log\log(|X|))$, where we omit $\eps$ and $\delta$ when they are equal to $1/4$.

\cite{BeimelNS:13approx} consider PAC learning with {\em approximate} $(\alpha,\beta)$-differential privacy where the privacy guarantee holds with probability $1-\beta$ (the basic notion is also referred to as {\em pure} to distinguish it from the approximate version). They show that $\Thr_b$ can be PAC learned using $O(16^{\log^*(b)} \cdot \log(1/\beta))$ samples ($\alpha$ is a constant as before). Their algorithm is proper so this separates the sample complexity of pure differentially private proper PAC learning from the approximate version. This work leaves open the question of whether such a separation can be proved for (non-proper) PAC learning.
\subsection{Our results}
In this paper we resolve the open problems described above.  In the process we also establish a new relation between $\SCDP$ and Littlestone's dimension, a well-studied measure of sample complexity of online learning \citep{Littlestone:87} (see \sectionref{sec:littlestone-def} for the definition).
The main ingredient of our work is a characterization of $\DRDim$ and $\PRDim$ in terms of randomized one-way
communication complexity of associated evaluation problems \citep{KremerNR:99}.
In such a problem Alice is given as input a function $f \in C$ and Bob is given an
input $x \in X$. Alice sends a single message to Bob, and Bob's goal is
to compute $f(x)$. The question is how many bits Alice must
communicate to Bob in order for Bob to be able to compute $f(x)$
correctly, with probability at least $2/3$ over the randomness used by Alice and Bob.

In the standard or ``private-coin'' version of this model, Alice and Bob each
have their own source of random coins. The minimal number of bits
needed to solve the problem for all $f\in C$ and $x \in X$ is denoted by $\Rand^{\rightarrow}(C)$.
In the stronger  ``public coin'' version of the model, Alice and Bob share the access to the same source of
random coins. The minimal number of bits needed to evaluate $C$ (with probability at least $2/3$) in this setting
is denoted by $\Rand^{\rightarrow,\pub}(C)$. See \sectionref{sec:comm-comp} for formal definitions.

We show that these communication problems are equivalent to
deterministic and probabilistic representation dimensions of $C$ and, in particular, $\SCDP(C) = \theta(\Rand^{\rightarrow,\pub}(C))$ (for clarity we omit the accuracy and confidence parameters, see \theoremref{thr:prdim-ccpub} and \theoremref{thr:drdim-cc} for details).
\begin{theorem}
\label{thm:main-eq-intro}
  $\DRDim(C) = \Theta(\Rand^{\rightarrow}(C))$ and
  $\PRDim(C) = \Theta(\Rand^{\rightarrow,\pub}(C))$.
\end{theorem}

The evaluation of threshold functions on a (discretized) interval $I_b$ corresponds to the well-studied ``greater than" function in communication complexity denoted as $\GT$. $\GT_b(x,y) = 1$ if and only if $x > y$,
where $x,y \in \zo^b$ are viewed as binary representations of integers. It is known
that $\Rand^{\rightarrow,\pub}(\GT_b) = \Omega(b)$ \citep{MiltersenNSW:98}. By combining this lower bound with
\theoremref{thm:main-eq-intro} we obtain a class whose $\VCDim$ dimension is 1 yet it requires at least $\Omega(b)$ samples to PAC learn differentially privately.

This equivalence also shows that some of the known results in \citep{BeimelKN:10, BeimelNS:13} are implied by well-known results from communication complexity, sometimes also giving simpler proofs. For example (1) the constant upper bound on the sample complexity of single points follows from the communication complexity of the equality function and (2) the bound $\DRDim(C) = O(\PRDim(C) + \log\log(|X|))$ follows from the classical result of \cite{Newman:91} on the relationship between the public and private coin models. See \sectionref{sec:cc-equiv-app} for more details and additional examples.

Our second contribution is a relationship of $\SCDP(C)$ (via the equivalences with $\Rand^{\rightarrow,\pub}(C)$) to Littlestone's \citeyearpar{Littlestone:87} dimension of $C$. Specifically, we prove
\begin{theorem}
\label{thm:ldinformal}
\begin{enumerate}
\item $\Rand^{\rightarrow,\pub}(C) = \Omega(\LDim(C))$.
\item For any $t$, there exists a class $C$ such that $\LDim(C)=2$ but $\Rand^{\rightarrow,\pub}(C) \geq t$.
\end{enumerate}
\end{theorem}
The first result follows from a natural reduction to the augmented index problem, which
is well-studied in communication complexity \citep{Bar-YossefJKK:04}. While new in our context, the relationship of Littlestone's dimension to quantum communication complexity was shown by \cite{Zhang:11}. Together with numerous known bounds on $\LDim$ \citep[\eg][]{Littlestone:87,MaassTuran:94b}, our result immediately yields a number of new lower bounds on $\SCDP$. In particular, results of \cite{MaassTuran:94b} imply that linear threshold functions over $I_b^d$ require $\Omega(d^2 \cdot b)$ samples to learn differentially privately. This implies that differentially private learners need to pay an additional dimension $d$ factor as well as a bit complexity of point representation $b$ factor over non-private learners. To the best of our knowledge such strong separation was not known before for problems defined over i.i.d. samples from a distribution (as opposed to worst case inputs). Note that this lower bound is also almost tight since $\log|\HS_b^d| =O(d^2(\log d + b))$ \citep[\eg][]{Muroga:71}.

In the second result of  \theoremref{thm:ldinformal} we use the class $\Line_p$ of lines in $\Z_p^2$ (a plane over a finite field $\Z_p$). A lower bound on the one-way quantum communication complexity of this class was first given by \cite{Aaronson:04} using his method based on a trace distance.

Finally, we consider PAC learning with $(\alpha,\beta)$-differential privacy. Our lower bound of $\Omega(b)$ on $\SCDP$ of thresholds together with the upper bound of $O(16^{\log^*(b)} \cdot \log(1/\beta))$ from \citep{BeimelNS:13approx} immediately imply a separation between the sample complexities of pure and approximate differential privacy. We show a stronger separation for the concept class $\Line_p$:
\begin{theorem}
  The sample complexity of $(\alpha, \beta)$-differentially privately learning $\Line_p$ is $O(\frac{1}{\alpha} \log(1/\beta))$.
\end{theorem}
Our upper bound is also simpler than the upper bound in \citep{BeimelNS:13approx}. See \sectionref{sec:impure-summary} for details.

\remove{
Finally, we also investigate distribution-specific private learning via distribution-specific analogs of $\PRDim$ and $\DRDim$.  We show that in this case probabilistic and deterministic representation are essentially the same thing, and that they correspond to the well-known notion of \emph{metric entropy}.  We give similar characterizations of these notions via communication complexity.  We also use these notions reprove known results about label privacy and distribution-specific privacy \citep{BeimelNS:13approx, ChaudhuriDKMT:06}.  We refer the reader to \sectionref{sec:dist} for the details.
}
\subsection{Related work}
There is now an extensive amount of literature on differential privacy in machine learning and related areas which we cannot hope to cover here. The reader is referred to the excellent surveys in \citep{SarwateC:13,DworkRoth:14}.

\cite{BlumDMN:05} showed that algorithms that can be implemented in the statistical query (SQ) framework of \cite{Kearns:98} can also be easily converted to differentially private algorithms. This result implies polynomial upper bounds on the sample (and computational) complexity of all learning problems that can be solved using statistical queries (which includes the vast majority of problems known to be solvable efficiently).
Formal treatment of differentially private PAC and agnostic learning was initiated in the seminal work of \cite{KLNRS:11}. Aside from the results we already mentioned, they separated SQ learning from differentially private learning. Further, they showed that SQ learning is (up to polynomial factors) equivalent to {\em local} differential privacy a more stringent model in which each data point is privatized before reaching the learning algorithm.

The results of this paper are for the distribution-independent learning, where the learner does not know the distribution over the domain. Another commonly-considered setting is distribution-specific learning in which the learner only needs to succeed with respect to a single fixed distribution $\D$ known to the learner. Differentially private learning in this setting and its relaxation in which the learner only knows a distribution close to $\D$ were studied by \cite{ChaudhuriHsu:11}. $\DRDim_\eps(C)$ restricted to a fixed distribution $\D$ is denoted by $\DRDim_\eps^\D(C)$ and equals to the logarithm of the smallest $\eps$-cover of $C$ with respect to the disagreement metric given by $\D$ (also referred to as the {\em metric entropy}). The standard duality between packing and covering numbers also implies that $\PRDim^\calD_{\frac{\eps}{2},\delta}(\calC) \geq \DRDim^\calD_{\eps}(\calC)  - \log(\tfrac{1}{1-\delta})$, and therefore these notions are essentially identical. It also follows from the prior work \citep{KLNRS:11,ChaudhuriHsu:11}, that $\DRDim_\eps^\D(C)$ characterizes the complexity of differentially private PAC and agnostic learning up to the dependence on the error parameter $\eps$ in the same way as it does for (non-private) learning \citep{BenedekI:91}. Namely,  $\Omega(\DRDim^\calD_{2\eps}(\calC)/\alpha)$ samples are necessary to learn $\alpha$-differentially privately with error $\eps$ (and even if only weaker label differentially-privacy is desired \citep{ChaudhuriHsu:11})and $O(\DRDim^\calD_{\eps/2}(\calC)/(\eps\alpha))$ samples suffice for $\alpha$-differentially private PAC learning. This implies that in this setting there are no dimension or bit-complexity costs incurred by differentially private learners. \cite{ChaudhuriHsu:11} also show that doubling dimension at an appropriate scale can be used to give upper and lower bounds on sample complexity of distribution-specific private PAC learning that match up to logarithmic factors.

In a related problem of sanitization of queries from the concept class $C$ the input is a database $D$ of points in $X$ and the goal is to output differentially privately a ``synthetic" database $\hat{D}$ such that for every $f\in C$, $\left|\frac{1}{|D|}\sum_{x \in D} f(x) - \frac{1}{|\hat{D}|}\sum_{x \in \hat{D}} f(x)\right| \leq \eps$. This problem was first considered by \cite{BlumLR:13} who showed an upper bound of $O(\VCDim(C) \cdot \log(|X|))$ on the size of the database sufficient for this problem and also showed a lower bound of $\Omega(b)$ on the number of samples required for solving this problem when $X=I_b$ for $C=\Thr_b$. It is easy to see that from the point of view of sample complexity this problem is at least as hard as (differentially private) {\em proper} agnostic learning of $C$ \citep[\eg][]{GuptaHRU:11}. Therefore
lower bounds on proper learning such as those in \citep{BeimelKN:10} and \citep{ChaudhuriHsu:11} apply to this problem and can be much larger than $\SCDP$ that we study. That said, to the best of our knowledge, the lower bound for linear threshold functions that we give was not known even for this harder problem.
Aside from sample complexity this problem is also computationally intractable for many interesting classes $C$ (see \citep{Ullman:13} and references therein for recent progress).

Sample complexity of more general problems in statistics was investigated in several works starting with \cite{DworkLei:09} (measured alternatively via convergence rates of statistical estimators) \citep{Smith:11,ChaudhuriH:12,DuchiJW:13focs,DuchiWJ13:nips}. A recent work of \cite{DuchiJW:13focs} shows a number of $d$-dimensional problems where differentially private algorithms must incur an additional factor $d/\alpha^2$ cost in sample complexity. However their lower bounds apply only to a substantially more stringent local model of differential privacy and are known not to hold in the model we consider here.

Differentially private communication protocols were studied by \cite{MPRTV:10} who showed that differential-privacy can be exploited to obtain a low-communication protocol and vice versa. Conceptually this result is similar to the characterization of sample complexity using $\PRDim$ given in \citep{BeimelNS:13}. Our contribution is orthogonal to \citep{MPRTV:10} since the main step in our work is going from a learning problem to a communication protocol for a different problem.

\section{Preliminaries}
\label{sec:prelims}

\subsection{Learning models}
\label{sec:learn-model}
\begin{definition}
\label{def:pac-model}
An algorithm $A$ PAC learns a concept class $C$ from $n$ examples if for every $\epsilon > 0, \delta>0$, $f\in C$ and distribution $\calD$ over $X$, $A$ given access to $S = \{(x_i, \ell_i)\}_{i \in [n]}$ where each $x_i$ is drawn randomly from $\calD$ and $\ell_i = f(x_i)$, outputs, with probability at least $1-\delta$ over the choice of $S$ and the randomness of $A$, a hypothesis $h$ such that $\pr_{x\sim \calD}[f(x) \neq h(x)] \leq \eps$.
\end{definition}
\noindent
{\bf Agnostic learning:} The {\em agnostic} learning model was introduced by \cite{Haussler:92} and \cite{KearnsSS:94} in order to model situations in which the assumption that examples are labeled by some $f \in C$ does not hold. In its least restricted version the examples are generated from some unknown distribution $P$ over $X \times \zo$. The goal of an agnostic learning algorithm for a concept class $C$ is to produce a hypothesis whose error on examples generated from $P$ is close to the best possible by a concept from $C$. For a Boolean function $h$ and a distribution $P$ over  $X \times \zo$ let $\Delta(P,h) = \pr_{(x,\ell) \sim P}[h(x) \neq \ell]$. Define $\Delta(P,C) = \inf_{h \in  C}\{\Delta(P,h)\}$.
\cite{KearnsSS:94} define agnostic learning as follows.
\begin{definition}
\label{def:agnostic-model}
An algorithm $A$ {\em agnostically} learns a concept class $C$ if for every $\epsilon > 0, \delta>0$, distribution $P$ over $X \times \zo$, $A$, given access to $S = \{(x_i, \ell_i)\}_{i \in [n]}$ where each $(x_i,\ell_i)$ is drawn randomly from $P$, outputs, with probability at least $1-\delta$ over the choice of $S$ and the randomness of $A$, a hypothesis $h$ such that $\Delta(P,h) \leq \Delta(P,C) + \eps$.
\end{definition}
In both PAC and agnostic learning model an algorithm that outputs a hypothesis in $C$ is referred to as {\em proper}.

\subsection{Differentially Private Learning}
\label{sec:prelims-dp}
Two sample sets $S = \{(x_i, \ell_i)\}_{i \in [n]}, S' = \{(x'_i,
      \ell'_i)\}_{i \in [n]}$ are said to be \emph{neighboring} if there exists $i \in
      [n]$ such that $(x_i, \ell_i) \neq (x'_i, \ell'_i)$, and for all $j
      \neq i$ it holds that $(x_j, \ell_j) = (x'_j, \ell'_j)$.  For $\al,\be>0$, an algorithm $A$ is
      $(\alpha,\beta)$-differentially private if for all neighboring $S, S' \in (X \times \zo)^n$ and for all $T \subseteq \mathsf{Range}(A)$:
      $$\Pr[A(S) \in T] \leq e^\alpha \Pr[A(S') \in T] + \beta,$$ where the probability is over the randomness of $A$ \citep{DworkMNS:06}. When $A$ is $(\alpha, 0)$-differentially private we say that it
      satisfies \emph{pure} differential privacy, which we also write
      as $\alpha$-differential privacy.


Intuitively, each sample $(x_i, \ell_i)$ used by a learning algorithm is the record of one individual, and
the privacy definition guarantees that by changing one record the output
distribution of the learner does not change by much.  We remark that,
in contrast to the accuracy of learning requirement, the differential privacy
requirement holds \emph{in the worst case} for all neighboring sets of examples $S, S'$,
not just those sampled i.i.d. from some distribution.  We refer the reader to the literature for a further justification of this notion of privacy \citep{DworkMNS:06}.

\remove{
\begin{definition}
  \label{def:dp}
  Let $\calC$ be a concept class of functions $X \dans \zo$.  A
  learning algorithm $A$ is said to learn a concept class $\calC$ with
  $(\alpha,\beta)$-differential privacy and $(\eps,\delta)$-accuracy
  using $n$ labeled examples if the following holds:
  \begin{description}
    \item[Accuracy:] For any $f \in \calC$ and any input distribution
      $\calD$ over $X$, it holds with probability $\geq 1 - \delta$ over the
      choice of $S = \{(x_i, y_i)\}_{i \in [n]}$ where each $x_i
      \drawnr \calD$ and $y_i = f(x_i)$ that $A(S)$ outputs $h$ such
      that $\Pr_{x \drawnr \calD}[h(x) = f(x)] \geq 1 - \eps$.
  \end{description}
\end{definition}
}

The \emph{sample complexity} $\SCDP_{\alpha,\eps,\delta}(\calC)$
is the minimal $n$ such that it is information-theoretically possible
to $(\eps, \delta)$-accurately and $\alpha$-differentially privately PAC learn $\calC$ with $n$ examples. $\SCDP$ without subscripts refers to $\SCDP_{1, \frac{1}{4}, \frac{1}{4}}$.

\subsection{Representation Dimension}
\begin{definition}[\citealp{BeimelKN:10}]
A class of functions $H$ \emph{ $\eps$-represents} $C$ if for every $f \in C$
and every distribution $\calD$ over the input domain of $f$, there
exists $h \in H$ such that $\Pr_{x\sim \calD}[f(x) \neq h(x)] \leq
\eps$. The \emph{deterministic representation dimension} of $C$, denoted as
$\DRDim_\eps(C)$ equals $\log(|H|)$ for the smallest $H$ that
$\eps$-represents $C$.  We also let $\DRDim(C) =
\DRDim_{\frac{1}{4}}(C)$.
\end{definition}
\begin{definition}[\citealp{BeimelNS:13}]
 A distribution $\calH$ over sets of boolean functions on $X$ is said
 to \emph{ $(\eps,\delta)$-probabilistically represent} $C$ if for
 every $f \in C$ and distribution $\calD$ over $X$, with probability
 $1-\delta$ over the choice of $H \drawnr \calH$, there exists $h \in
 H$ such that $\Pr_{x \sim \calD}[h(x) \neq f(x)] \leq \eps$.  The
 \emph{$(\eps,\delta)$-probabilistic representation dimension}
 $\PRDim_{\eps,\delta}(C)$ equals the minimal value of $\max_{H \in
   \supp(\calH)} \log |H|$, where the minimum is over all $\calH$ that
 $(\eps,\delta)$-probabilistically represent $C$.  We also let
 $\PRDim(C) = \PRDim_{\frac{1}{4}, \frac{1}{4}}(C)$.
\end{definition}

\cite{BeimelNS:13} proved the following characterization of $\SCDP$ by $\PRDim$.
\begin{theorem}[\citealp{KLNRS:11,BeimelNS:13}]
  \label{thm:scdp-equiv}
  \begin{align*}
    \SCDP_{\alpha,\eps,\delta}(\calC) & =
    O\left(\frac{1}{\alpha \eps} \left(\log(1/\eps)\cdot \left(\PRDim_{\frac{1}{4},\frac{1}{4}}(\calC) + \log\log\frac{1}{\eps\delta}\right) + \log \frac{1}{\delta}\right) \right) . \\
    \SCDP_{\alpha,\eps,\delta}(\calC) & =
    \Omega\left(\frac{1}{\alpha \eps} \PRDim_{1/4,1/4}(\calC) \right) .
  \end{align*}
  For agnostic learning we have that sample complexity is at most
  $$O\left(\left(\frac{1}{\al\eps}+\frac{1}{\eps^2}\right) \left(\log(1/\eps)\cdot \left(\PRDim_{\frac{1}{4},\frac{1}{4}}(\calC) + \log\log\frac{1}{\eps\delta}\right) + \log \frac{1}{\delta}\right) \right).$$
\end{theorem}
This form of upper bound combines accuracy and confidence boosting from \citep{BeimelNS:13} to first obtain $(\eps,\delta)$-probabilistic representation and then the use of exponential mechanism as in \citep{KLNRS:11}. The results in \citep{KLNRS:11} show the extension of this bound to agnostic learning. Note that the characterization for PAC learning is tight up to logarithmic factors.

\subsection{Communication Complexity}
\label{sec:comm-comp}
Let $X$ and $Y$ be some sets. A private-coin one-way protocol $\pi(x, y)$ from Alice who holds $x \in X$ to Bob who holds $y \in Y$ is given by Alice's randomized algorithm producing a communication $\sigma$ and Bob's
randomized algorithm which outputs a boolean value. We describe Alice's algorithm
by a function $\pi_A(x; r_A)$ of the input $x$ and random bits and Bob's
algorithm $\pi_B(\sigma, y; r_B)$ by a function of input $y$,
communication $\sigma$ and random bits. (These algorithms need not be
efficient.)  The (randomized) output of the
protocol on input $(x,y)$ is the value of $\pi(x,y;r_A,r_B) \triangleq \pi_B(\pi_A(x; r_A),y; r_B)$ on a randomly and uniformly chosen $r_A$ and $r_B$. The cost of the protocol $\CC(\pi)$ is given by the maximum
$|\sigma|$ over all $x \in X, y \in Y$ and all
possible random coins.

A public-coin one-way protocol $\pi(x, y)$ is given by a randomized Alice's
algorithm described by a function $\pi_A(x; r)$ and a randomized Bob's algorithm described by a function $\pi_B(\sigma, x; r)$.
The (randomized) output of the
protocol on input $(x,y)$ is the value of $\pi(x,y;r) \triangleq \pi_B(\pi_A(x; r),y; r)$ on a randomly and uniformly chosen $r$. The cost of the protocol $\CC(\pi)$ is
defined as in the private-coin case.

Let $\Pi^{\rightarrow}_\eps(g)$ denote the class of all private-coin
one-way protocols $\pi$ computing $g$ with error $\eps$, namely
private-coin one-way protocols $\pi$ satisfying for all $x \in
X, y \in Y$
$$\Pr_{r_A, r_B}[\pi(x, y; r_A, r_B) = g(x, y)] \geq 1 - \eps .$$
Define $\Pi^{\rightarrow,\pub}_\eps(g)$ similarly as the class of all
public-coin one-way protocols $\pi$ computing $g$.  Define
$\Rand_\eps^{\rightarrow}(g) = \min_{\pi \in
  \Pi^{\rightarrow}_\eps(g)} \CC(\pi)$ and
$\Rand_\eps^{\rightarrow,\pub}(g) = \min_{\pi \in \Pi^{\rightarrow,\pub}_\eps(g)} \CC(\pi)$.

A deterministic one-way protocol $\pi$ and its cost are defined as above but without dependence on random bits.
We will also require distributional notions of complexity, where there
is a fixed input distribution from which $x, y$ are drawn.  For a distribution $\mu$ over $X \times Y$, we define
$\Pi^{\rightarrow}_\eps(g; \mu)$ to be all deterministic one-way
protocols $\pi$ such that
$$\Pr_{(x,y) \sim \mu}[\pi(x, y) = g(x, y)] \geq 1 - \eps .$$
Define $\Dist^{\rightarrow}_\eps(g; \mu) = \min_{\pi \in
  \Pi^{\rightarrow}_\eps(g;\mu)} \CC(\pi)$.  A standard averaging
argument shows that the quantity $\Dist^{\rightarrow}_\eps(g;\mu)$ remains
unchanged even if we took the minimum over randomized (either public
or private coin) protocols computing $g$ with error $\leq \eps$ (\ie
since there must exist a fixing of the private coins that achieves as
good error as the average error).

Yao's minimax principle \citep{Yao:1977} states that for all functions $g$:
\begin{equation}
  \label{eq:minmax}
  \Rand^{\rightarrow,\pub}_\eps(g) = \max_{\mu} \Dist^{\rightarrow}_\eps(g; \mu) .
\end{equation}

Error in both public and private-coin protocols can be reduced by using several independent copies of the protocol and then taking a majority vote of the result. This implies that for every $\eps,\gamma \in (0,1/2)$,
\begin{equation}
\Rand^{\rightarrow,\pub}_\eps(f) = O(\Rand^{\rightarrow,\pub}_{1/2-\gamma}(f) \cdot \log(1/\eps)/\gamma^2)\label{eq:erroramplify} .
\end{equation}
Analogous statement holds for $\Rand^{\rightarrow}$. This allows us to treat protocols with constant errors in $(0,1/2)$ range as equivalent up to a constant factor in the communication complexity.

\subsection{Littlestone's Dimension}
\label{sec:littlestone-def}
While in this work we will not use the definition of the online mistake-bound model itself, we briefly describe it for completeness.
In the online mistake-bound model learning proceeds in rounds. At the beginning of round $t$, a learning algorithm has some hypothesis $h_t$. In round $t$, the learner sees a point $x_t \in X$ and predicts $h_t(x_t)$. At the end of the round, the correct label $y_t$ is revealed and the learner makes a mistake if $h_t(x_t)\neq y_t$. The learner then updates its hypothesis to $h_{t+1}$ and this process continues. When learning a concept class $C$ in this model $y_t = f(x_t)$ for some unknown $f \in C$. The (sample) complexity of such learning is defined as the largest number of mistakes that any learning algorithm can be forced to make when learning $C$. \cite{Littlestone:87} proved that it is exactly characterized by a dimension defined as follows.

Let $C$ be a concept class over domain $X$. A {\em mistake tree} $T$ over $X$ and $C$ is a binary tree in which each internal node $v$ is labelled by a point $x_v \in X$ and each leaf $\ell$ is labelled by a concept $c_\ell \in C$. Further, for every node $v$ and leaf $\ell$: if $\ell$ is in the right subtree of $v$ then $c_\ell(x_v) = 1$, otherwise $c_\ell(x_v) = 0$. We remark that a mistake tree over $X$ and $C$ does not necessarily include all concepts from $C$ in its leaves. Such a tree is called {\em complete} if all its leaves are at the same depth. Littlestone's dimension $\Ldim(C)$ is defined as the depth of the deepest complete mistake tree $T$ over $X$ and $C$ \citep{Littlestone:87}. Littlestone's dimension is also known to exactly characterize the number of (general) equivalence queries required to learn $C$ in Angluin's \citeyearpar{Angluin:88} exact model of learning \citep{Littlestone:87}.

\section{Equivalence between representation dimension and
  communication complexity}
\label{sec:cc-equiv}
We relate communication complexity to private learning by considering the communication problem associated with evaluating a function $f$ from a concept class $\calC$ on an input $x \in X$.
Formally, for a Boolean concept class $\calC$ over domain $X$, define $\Eval_\calC : \calC
\times X \dans \zo$ to be the function defined as $\Eval_\calC(f, x) =
f(x)$. In a slight abuse of notation we use $\Rand^{\rightarrow,\pub}_{\eps}(\calC)$ to denote $\Rand^{\rightarrow,\pub}_{\eps}(\Eval_\calC)$ (and similarly for $\Rand^{\rightarrow}_{\eps}(\calC)$).


Our main result is the following two bounds.
\begin{theorem}
\label{thr:prdim-ccpub}
  For any $\eps \in [0, 1/2]$ and $\delta \in [0, 1]$, and any concept
  class $\calC$, it holds that:
  \begin{itemize}
  \item $\PRDim_{\eps,\delta}(\calC) \leq
    \Rand^{\rightarrow,\pub}_{\eps\delta}(\calC)$.
  \item $\PRDim_{\eps,\delta}(\calC) \geq
    \Rand^{\rightarrow,\pub}_{\eps+\delta-\eps\delta}(\calC)$.
  \end{itemize}
\end{theorem}
\begin{proof}
  \noindent $(\leq)$: let $\pi$ be the public-coin one-way protocol
  that achieves the optimal communication complexity $c$.  For each
  choice of the public random coins $r$, let $H_r$ denote the set of
  functions $h_\sigma(x) = \pi_B(\sigma, x; r)$ over all possible
  $\sigma$.  Thus, each $H_r$ has size at most $2^c$.  Let the
  distribution $\calH$ be to choose uniformly random $r$ and then
  output $H_r$.

  We show that this family $(\eps,\delta)$-probabilistically
  represents $\calC$.  We know from the fact that $\pi$ computes
  $\Eval_\calC$ with error $\eps\delta$ that it must hold for all $f
  \in \calC$ and $x \in X$ that:
  $$\Pr_{r}[\pi_B(\pi_A(f; r), x; r) \neq f(x)]
  \leq \eps \delta .$$
  In particular, it must hold for any distribution $\calD$ over
  $X$ that:
  $$\Pr_{x \sim \calD, r}[\pi_B(\pi_A(f; r), x; r) \neq f(x)]
  \leq \eps \delta .$$
  Therefore, it must hold that
  $$\Pr_{r} \left[ \Pr_{x \sim\calD}
    [\pi_B(\pi_A(f; r), x; r) \neq f(x)] > \eps\right] < \delta .$$
  Note that $\pi_B(\pi_A(f; r), x; r) \equiv h_{\pi_A(f; r)}(x) \in H_r$ and therefore, with probability $\geq 1 - \delta$ over the choice of
  $H_r \drawnr \calH$, there exists $h \in H_r$ such that $\Pr_{x \sim\calD}[h(x) \neq
  f(x)] \leq \eps$.

  \noindent $(\geq)$: let $\calH$ be the distribution over sets of
  boolean functions that achieves $\PRDim_{\eps,\delta}(\calC)$.
  We will show that for each distribution $\mu$ over inputs $(f, x)$,
  we can construct a $(\eps+\delta-\eps\delta)$-correct protocol for $\Eval_\calC$
  over $\mu$ that has communication bounded by
  $\PRDim_{\eps,\delta}(\calC)$.  Namely, we will prove that
  \begin{equation}
    \label{eq:distbound}
    \max_\mu \Dist^{\rightarrow}_{\eps+\delta-\eps\delta}(\Eval_\calC; \mu) \leq
    \PRDim_{\eps,\delta}(\calC) .
  \end{equation}
  By Yao's minimax principle (\equationref{eq:minmax}) \citep{Yao:1977} this implies that
  $$\Rand^{\rightarrow,\pub}_{\eps+\delta-\eps\delta}(\calC) \leq \PRDim_{\eps,\delta}(\calC) .$$

  Fix $\mu$.  This induces a marginal distribution $\calF$ over
  functions $f \in \calC$ and for every $f\in \calC$ a distribution
  $\calD_f$ which is $\mu$ conditioned on the function being $f$ (note
  that $\mu$ is equivalent to drawing $f$ from $\calF$ and then $x$
  from $\calD_f$).  The protocol $\pi$ is defined as follows: use
  public coins to sample $H \drawnr \calH$.  Alice knows $f$ and so
  knows the distribution $\calD_f$.  Alice sends the index of $h \in
  H$ such that $\Pr_{x \sim\calD_f}[h(x) \neq f(x)] \leq \eps$ if such $h$
  exists or an arbitrary $h \in H$ otherwise. Bob returns $h(x)$.

  The error of this protocol can be analyzed as follows.  Fix $f$ and
  let $G_f$ denote the event that $H \drawnr \calH$ contains $h$ such
  that $\Pr_{x \sim \calD_f}[h(x) \neq f(x)] \leq \eps$.
  Observe that $G_f$ is independent of $\calD_f$ so that
  even conditioned on $G_f$ $x$ remains distributed according to
  $\calD_f$. Also, since $\calH$ $(\eps,\delta)$-probabilistically
  represents $\calC$, we know that for every $f$, $\Pr_r[G_f] \geq 1 -
  \delta$.  Therefore we can then deduce that:
  \begin{align*}
    \Pr_{r, (f,x) \sim \mu}[\pi(f,x;r) = f(x)] & =
    \Pr_{r, (f,x) \sim \mu}[\pi(f,x;r) = f(x) \wedge G_f] + \Pr_{r,
      (f,x) \sim \mu} [\pi(f,x;r) = f(x) \wedge \neg G_f] \\
    & \geq \Pr_{r, f \sim \calF}[G_f] \cdot \Pr_{r,    x\sim   \calD_f}[\pi(f,x;r) = f(x) \mid G_f] \\
    & \geq (1 - \delta)(1-\eps)  = 1 - \delta - \eps + \epsilon\delta .
  \end{align*}
  Thus $\pi$ computes $\calC$ with error at most $\eps + \delta-\eps\delta$ and
  it has communication bounded by $\PRDim_{\eps,\delta}(\calC)$.
\end{proof}

We also establish an analogous equivalence for $\DRDim$ and private-coin protocols.
\begin{theorem}
\label{thr:drdim-cc}
  For any $\eps \in [0, 1/2]$, it holds that:
  \begin{itemize}
  \item $\DRDim_\eps(\calC) \leq
    \Rand_{\eps/2}^{\rightarrow}(\calC)$ .
    \item $\DRDim_\eps(\calC) \geq
    \Rand_{\eps}^{\rightarrow}(\calC)$ .
  \end{itemize}
\end{theorem}

\begin{proof}
  \noindent $(\leq)$: let $\Rand_{\eps/2}^{\rightarrow}(\calC) = c$
  and fix the private-coin one-way protocol $\pi$ that achieves $c$.
  We define the deterministic representation $H$ to be all functions
  $h_\sigma(x) = \mathsf{maj}_{r_B} \{ \pi(\sigma, x; r_B) \}$, \ie the
  majority value of Bob's outputs on input $x$ and
  communication $\sigma$.  Observe that there are $2^c$ such functions
  (one for each $\sigma$ possible) and therefore it suffices to show
  that $H$ $\eps$-deterministically represents $\calC$.  To see this,
  observe that for each $f \in \calC$, and all $x \in X$, it holds
  that:
  $$\Pr_{r_B, \sigma \drawnr \pi_A(f; r_A)}[f(x) = \pi_B(\sigma, x; r_B)]
  \geq 1 - \eps/2.$$
  In particular, this means that for all distributions $\calD$ over
  $X$, it holds that
  $$\Pr_{x\sim\calD, r_B, \sigma \drawnr \pi_A(f; r_A)}[f(x) =
  \pi_B(\sigma, x; r_B)] \geq 1 - \eps/2 .$$
  By a standard averaging argument, there must exist at least one
  $\sigma$ such that
  $$\Pr_{x \sim\calD, r_B}[f(x) = \pi_B(\sigma, x; r_B)]\geq 1 -
  \eps/2 .$$
  Now say that $x$ is bad if $\Pr_{r_B}[f(x) = \pi_B(\sigma, x; r_B)]
  < 1/2$.  By the above, it follows that $\Pr_{x \sim\calD}[x \text{ bad}] \leq
  \eps$.  By definition, if $x$ is not bad then $f(x) = h_\sigma(x)$,
  since $h_\sigma$ is the majority of $\pi_B(\sigma, x;
  r_B)$ over all $r_B$.  Therefore
  $$\Pr_{x\sim\calD}[f(x) = h_\sigma(x)] \geq 1 - \eps .$$
  This implies that $H$ $\eps$-deterministically represents $\calC$.

  \noindent $(\geq)$: We first apply von-Neumann's Minimax theorem to the
  definition of deterministic representation.  In particular, suppose
  $H$ is the family of functions that achieves
  $\DRDim_\eps(\calC)$. Thus, for each $f \in \calC$ and each
  distribution $\calD$ over $X$, there exists $h \in H$ such that
  $\Pr_{\calD}[h(x) = f(x)] \geq 1 - \eps$.  We define a zero-sum game for each $f$ with
  the first player choosing a point $x \in X$ and the second player
  choosing a hypothesis $h \in H$ and the payoff of the second player being $|h(x) - f(x)|$. The definition of $\DRDim_\eps(\calC)$ implies that for every mixed strategy of the first player the second player has a pure strategy that achieves payoff of at least $1-\eps$. By the Minimax theorem there exists
  a distribution $\calh_f$ over $H$ such that, for every $x \in
  X$, it holds that
  $$\E_{h \sim \calh_f}[|h(x) - f(x)|] = \Pr_{h \sim \calh_f}[h(x) = f(x)] \geq 1 - \eps .$$
  Our private-coin protocol $\pi$ for $\Eval_\calC$ will be the following:
  on input $f$, Alice will use her private randomness to sample $h
  \sim \calh_f$ and send the index of $h$ to Bob.  Bob then outputs $h(x)$.
  Thus, for each $f, x$, it holds that
  $$\Pr_{\pi}[\pi(f,x) = f(x)] = \Pr_{h \sim \calh_f}[h(x) =
  f(x)] \geq 1 - \eps$$
  and so the protocol computes $\Eval_\calC$ with error $\leq \eps$.
\end{proof}

An immediate corollary of these equivalences and eq.\eqref{eq:erroramplify} is that
  $\DRDim(C) = \Theta(\Rand^{\rightarrow}_{1/3}(\calC))$ and
  $\PRDim(C) = \Theta(\Rand^{\rightarrow,\pub}_{1/3}(\calC))$ as we stated in \theoremref{thm:main-eq-intro}.

\subsection{Applications}
\label{sec:cc-equiv-app}

Our equivalence theorems allow us to import many results from
communication complexity into the context of private PAC learning,
both proving new facts and simplifying proofs of previously known
results in the process.

\paragraph{Separating $\SCDP$ and $\VCDim$ dimension.}
Define $\Thr_b$ as the family of
functions $t_x : I_b \rightarrow \zo$ for $x \in I_b$ where
$t_x(y) = 1$ if and only if $y \geq x$. The lower bound follows from
an observation that $\Eval_{\Thr_b}$ is equivalent to the
``greater-than" function $\GT_b(x,y) = 1$ if and only if $x > y$,
where $x,y \in \zo^b$ are viewed as binary representations of integers
in $I_b$. Note $\Eval_{\Thr_b}(t_x, y) = 1-\GT_b(x, y)$ and
therefore these functions are the same up to the negation.  $\GT_b$ is
a well studied function in communication complexity and it is known
that $\Rand^{\rightarrow,\pub}_{1/3}(\GT_b) = \Omega(b)$
\citep{MiltersenNSW:98}. By combining this lower bound with
\theoremref{thr:prdim-ccpub} we obtain that $\VCDim(\Thr_b) = 1$ yet
$\PRDim(\Thr_b) = \Omega(b)$.  From \theoremref{thm:scdp-equiv} it
follows that $\SCDP(\Thr_b) = \Omega(b)$.

We note that it is known that VC dimension corresponds to the maximal
distributional one-way communication complexity over all \emph{product} input
distributions. Hence this separation is analogous to separation of
distributional one-way complexity over product distributions and the
maximal distributional complexity over all distributions achieved
using the greater-than function \citep{KremerNR:99}.

We also give more such separations using lower bounds on $\PRDim$
based on Littlestone's dimension.  These are discussed in \sectionref{sec:ldim}.

\paragraph{Accuracy and confidence boosting.}
Our equivalence theorems give a simple alternative way to reduce error in probabilistic and deterministic representations
without using sequential boosting as was done in \citep{BeimelNS:13}.
 Given a private PAC learner with constant error, say $(\eps, \delta) = (1/8, 1/8)$, one
can first convert the learner to a communication protocol with
error $1/4$, use $O(\log \frac{1}{\eps' \delta'})$ simple independent repetitions (as in eq.\eqref{eq:erroramplify}) to
reduce the error to $\eps' \delta'$, and then convert the protocol
back into a $(\eps', \delta')$-probabilistic
representation. The ``magic'' here happens when we convert
  between the communication complexity and probabilistic
  representation using min-max type arguments. This is
  the same tool that can be used to prove (computationally inefficient) boosting
  theorems.

\paragraph{Probabilistic vs. deterministic representation dimension.}
It was shown by \citet{Newman:91} that public and private coin
complexity are the same up to additive logarithmic terms.  In our
setting (and with a specific choice of error bounds to simplify
presentation), Newman's theorem implies that
\begin{equation}
  \label{eq:newman}
  \Rand^{\rightarrow}_{1/8}(\calC) \leq  \Rand^{\rightarrow,\pub}_{1/9}(\calC) +
  O(\log \log (|\calC| |X| )) .
\end{equation}
We know by Sauer's lemma that $\log |\calC| \leq O(\VCDim(\calC) \cdot
\log |X|)$,
therefore we deduce that:
$$\Rand^{\rightarrow}_{1/8}(\calC) \leq \Rand^{\rightarrow,\pub}_{1/9}(\calC) +
O(\log \log \VCDim(\calC) + \log \log |X| ) .$$ %
By our equivalence theorems, $\DRDim(\calC) = \DRDim_{1/4}(\calC) \leq \Rand^{\rightarrow}_{1/8}(\calC)$ and
$\Rand^{\rightarrow,\pub}_{1/9}(\calC) \leq \PRDim_{1/16,1/24}(\calC) = O(\PRDim(\calC))$.
This implies that
$$\DRDim(\calC) = O(\PRDim(\calC) 
 + \log\log|X|). $$ %
A version of this was first proved in
\citep{BeimelNS:13}, whose proof is similar in spirit to the proof of
Newman's theorem.  We also remark that the fact that
$\DRDim_{1/3}(\mathsf{Point}_b) = \Omega(\log b)$ while
$\PRDim_{1/3}(\mathsf{Point}_b) = O(1)$ \citep{BeimelKN:10,
  BeimelNS:13} corresponds to the fact that the private-coin
complexity of the equality function is $\Omega(\log b)$, while the
public-coin complexity is $O(1)$.  Here $\mathsf{Point}_b$ is the
family of point functions over $\zo^b$, \ie functions that are zero everywhere
except on a single point.

\paragraph{Simpler learning algorithms.} Using our equivalence
theorems, we can ``import'' results from communication complexity to
give simple private PAC learners.  For example, the well-known
constant communication equality protocol using inner-product-based hashing can be
converted to a probabilistic representation using
\theoremref{thr:prdim-ccpub}, which can then be used to learn point
functions.  The resulting learning algorithm is somewhat simpler than the constant sample
complexity learner for $\mathsf{Point}_b$ described in
\citep{BeimelKN:10} and we believe that this view also provides useful
intuition. We remark that the probabilistic representation for $\mathsf{Point}_b$ that results from the communication protocol
is known and was used for learning $\mathsf{Point}_b$ by \cite{Feldman:09robust} in the context of evolvability. A closely related representation is also mentioned in \citep{BeimelNS:13}.

Furthermore in some cases this connection can lead to
\emph{efficient} private agnostic learning algorithms.  Namely, if there is
a communication protocol for $\Eval_\calC$ where
Bob's algorithm is polynomial-time then one can run
the exponential mechanism in time $2^{O(\PRDim(C))}$ to differentially privately
agnostically learn $\calC$.  

 \remove{
  Curiously, this idea does not work for non-private
learning, since the communication-complexity characterization of VC
dimension is a necessarily average-case notion rather than worst-case,
and one cannot reduce error by simple repetition in the average-case
setting.}

\section{Lower Bounds via Littlestone's Dimension}
\label{sec:ldim}

In this section, we show that Littlestone's dimension lower bounds the sample complexity of differentially private learning.
Let $C$ be a concept class over $X$ of $\Ldim$ $d$. Our proof is based on a reduction from the communication complexity of $\Eval_{C}$ to
the communication complexity of Augmented Index problem on $d$ bits. $\AI$ is the promise problem where Alice gets a string $x_1,
  \ldots, x_d \in \zo^d$ and Bob gets $i \in [d]$ and $x_1, \ldots,
  x_{i-1}$, and $\AI(x, (i, x_{[i-1]})) = x_i$ where $x_{[i-1]} =
  (x_1, \ldots, x_{i-1})$. A variant of this problem in which the length of the prefix is not necessarily $i$ but some additional parameter $m$ was first explicitly defined by \citet{Bar-YossefJKK:04} who proved that it has randomized one-way communication complexity of $\Omega(d-m)$. The version defined above is from \citep{BaIPW:10} where it is also shown that a lower bound for $\AI$ follows from an earlier work of \citep{MiltersenNSW:98}.  We use the following lower bound for $\AI$.
  \begin{lemma}
\label{lem:lb-thr}
  $\Rand^{\rightarrow}_\eps(\AI) \geq (1 - \Ent(\eps)) d$, where $\Ent(\eps)=\eps \log(1/\eps) + (1-\eps) \log(1/(1-\eps))$ is the binary entropy function.
\end{lemma}
  A proof of this lower bound can be easily derived by adapting the proof in \citep{Bar-YossefJKK:04} and we include it in \sectionref{sec:ppac-cc-ai}.



We now show that if $\Ldim(C) = d$ then one can reduce $\AI$ on $d$ bit inputs to $\Eval_{C}$.
\begin{lemma}
\label{lem:ldim-ai}
Let $C$ be a concept class over $X$ and $d=\Ldim(C)$. There exist two mappings $m_C: \zo^d \rightarrow C$ and $m_X: \bigcup_{i \in [d]} \zo^i \rightarrow X$ such that for every $x$ and $i\in [d]$, the value of $m_C(x)$ on point $m_X(x_{[i-1]})$ is equal to $\AI(x, (i, x_{[i-1]})) = x_i$.
\end{lemma}
\begin{proof}
By the definition of $\Ldim$, there exists a complete mistake tree $T$ over $X$ and $C$ of depth $d$. Recall that a mistake tree over $X$ and $C$ is a binary tree in which each internal node is labelled by a point in $X$ and each leaf is labelled by a concept in $C$.
For $x \in \zo^d$ consider a path from the root of the tree such that at step $j \in [d]$ we go to the left subtree if $x_j=0$ and the right subtree if $x_j = 1$. Such path will end in a leaf which we denote by $\ell_x$ and the concept that labels it by $c_x$. For a prefix $x_{[i-1]}$, let $v_{x_{[i-1]}}$ denote the internal node at depth $i$ on this path (with $v_\emptyset$ being the root) and let $z_{x_{[i-1]}}$ denote the point in $X$ which labels $v_{x_{[i-1]}}$.

 We define the mapping $m_C$ as $m_C(x) = c_x$ for all $x \in \zo^d$ and the mapping $m_X$ as $m_X(y) = z_y$ for all $y \in \bigcup_{i \in [d]} \zo^i$.
By the definition of a mistake tree over $X$ and $C$, the value of the concept $c_x$ on the point $z_{x_{[i-1]}}$ is determined by whether the leaf $\ell_x$ is in the right (1) or the left (0) subtree of the node $v_{x_{[i-1]}}$. Recall that the turns in the path from the root of the tree to $\ell_x$ are defined by the bits of $x$. At the node $v_{x_{[i-1]}}$,  $x_i$ determines whether $\ell_x$ will be in the right or the left subtree. Therefore $c_x(z_{x_{[i-1]}}) = x_i$. Therefore the mapping we defined reduces $\AI$ to $\Eval_{C}$.
\end{proof}


An immediate corollary of \lemmaref{lem:ldim-ai} and \lemmaref{lem:lb-thr} is the following lower bound.
\begin{corollary}
\label{cor:ldim-to-rand}
Let $C$ be a concept class over $X$ and $d=\Ldim(C)$. $\Rand^{\rightarrow}_\eps(C) \geq (1 - \Ent(\eps)) d$.
\end{corollary}
A stronger form of this lower bound was proved by \citet{Zhang:11} who showed that the power of Partition Tree lower bound technique for one-way {\em quantum} communication complexity of \citet{Nayak:99} can be expressed in terms of $\Ldim$ of the concept class associated with the communication problem.

\subsection{Applications}
We can now use numerous known lower bounds for Littlestone's dimension of $C$ to obtain lower bounds on sample complexity of private PAC learning. Here we list several examples of known results where $\Ldim(C)$ is (asymptotically) larger than the VC dimension of $C$.
\begin{enumerate}
\item $\Ldim(\Thr_b) = b$ \citep{Littlestone:87}. $\VCDim(\Thr_b) = 1$.
\item Let $\BOX_b^d$ denote the class of all axis-parallel rectangles over $[2^b]^d$, namely all concepts $r_{s,t}$ for $s, t \in [2^b]^d$ defined as $r_{s,t}(x) = 1$ if and only if for all $i \in [d]$, $s_i \leq x_i\leq t_i$. $\Ldim(\BOX_b^d) \geq b \cdot d$ \citep{Littlestone:87}. $\VCDim(\BOX_b^d) = d+1$.
\item Let $\HS_b^d$ denote class of all linear threshold functions over $[2^b]^d$. $\Ldim(\HS_b^d) = b \cdot d(d-1)/2$. This lower bound is stated in \citep{MaassTuran:94}. We are not aware of a published proof and therefore a proof based on counting arguments in \citep{Muroga:71} appears in  \sectionref{sec:ldim-hs}  for completeness.
    $\VCDim(\HS_b^d) = d+1$.
\item Let $\BALL_b^d$ denote class of all balls over $[2^b]^d$, that is all functions obtained by restricting a Euclidean ball in $\R^d$ to $[2^b]^d$. Then $\Ldim(\BALL_b^d) = \Omega(b \cdot d^2)$ \citep{MaassTuran:94b}. $\VCDim(\BALL_b^d) = d+1$.
\end{enumerate}

\section{Separation of $\PRDim$ from $\LDim$}

We next consider the question of whether $\PRDim$ is equal to
$\Ldim$. It turns out that the communication complexity literature \citep{Zhang:11} already
contains the following counter-example separating $\PRDim$ and
$\Ldim$. Define:
$$\Line_p = \{ f : \Z_p^2 \dans \zo : \exists a,b
\in \Z_p^2\ s.t.\ f(x,y) = 1 \text{ iff }ax + b = y\}$$
It is easy to see that $\Ldim(\Line_p) = 2$ (an online learning algorithm only needs two different counterexamples to the constant 0 function to recover the unknown line). It was also shown
\citep{Aaronson:04} that the quantum one-way communication complexity of
$\Eval_{\Line_p}$ is $\Theta(\log p)$. This already implies a
separation between $\Ldim$ and $\PRDim$ using
\theoremref{thr:prdim-ccpub} and the fact that quantum one-way
communication lower-bounds randomized public-coin communication.

We give a new information-theoretic and simpler proof of Aaronson's result for randomized public-coin communication. We start with a brief review of basic notions from information theory.

\subsection{Information theory background}
\label{sec:it-bg}
We will use the convention of letting bold-face $\bfa, \bfb$ denote random
variables and regular type $a, b$ denote particular values that those
random variables may take.

Recall the following definitions of entropy (all logarithms are base $2$):
\begin{eqnarray*}
  \text{(Shannon entropy)}  & \Ent(\bfx)  = & \sum_x \Pr[\bfx = x]
  \log \tfrac{1}{\Pr[\bfx = x]} \\
  \text{(R\'{e}nyi entropy or collision entropy)}   & \Ent_2(\bfx)  =
  & \log \frac{1}{\sum_x \Pr[\bfx = x]^2} \\
  \text{(Min-entropy)}&   \Ent_\infty(\bfx)  =  & \min_x \log \tfrac{1}{\Pr[\bfx = x]}
\end{eqnarray*}
Recall that for all random variables $\bfx$ over some universe $X$, it
holds that $\log |X| \geq \Ent(\bfx) \geq \Ent_2(\bfx) \geq
\Ent_\infty(\bfx)$.

The conditional Shannon entropy is defined as $\Ent(\bfx \mid \bfy) =
\Exp_{y \drawnr \bfy} \left[ \Ent(\bfx \mid \bfy = y) \right]$.

The (Shannon) mutual information is defined as $\Info(\bfx; \bfy \mid
\bfz) = \Ent(\bfx \mid \bfz) - \Ent(\bfx \mid \bfy \bfz)$.
The mutual information satisfies the \emph{chain rule}:
$$ \Info(\bfx; \bfy \bfz \mid \bfm) = \Info(\bfx; \bfy \mid \bfm) +
\Info(\bfx; \bfz \mid \bfy \bfm)\ .$$


The Kullback-Leibler divergence (also called relative entropy) is
defined as:
$$\Div(\bfx \parallel \bfx') = \sum_x \Pr[\bfx = x] \log \frac{\Pr[\bfx
  = x]}{\Pr[\bfx' = x]} .$$

For two jointly distributed random variables $\bfx \bfy$, let
$\angles{\bfx} \angles{\bfy}$ denote independent samples from the
marginal distributions of $\bfx$ and $\bfy$.  If conditioned on some
event $E$, we write $(\angles{\bfx} \angles{\bfy} \mid E) =
\angles{\bfx \mid E} \angles{\bfy \mid E}$.

Recall the following characterization of mutual information in terms
of Kullback-Leibler divergence:
\begin{equation}
  \Info(\bfx; \bfy) = \Div(\bfx \bfy \parallel \angles{\bfx} \angles{\bfy}) .
\end{equation}

\begin{lemma}
  \label{lem:min-ent-compress}
  Let $\bfx \bfy$ be jointly distributed random variables.  Suppose
  the support of $\bfy$ has size $2^s$.  Then for every $t > 0$ it holds
  that:
  $$\Pr_{y \drawnr \bfy}[\Ent_\infty(\bfx \mid \bfy = y) <
  \Ent_\infty(\bfx) - s - t] < 2^{-t} .$$
\end{lemma}
\begin{proof}
  Let $k = \Ent_\infty(\bfx)$.  Say that $y$ is \emph{bad} if
  $\Ent_\infty(\bfx \mid \bfy = y) < k - s - t$.  By Bayes' rule
  and the definition of min-entropy, it holds that for all bad $y$
  there exists $x$ such that
  $$\frac{2^{-k}}{\Pr[\bfy = y]} \geq \frac{\Pr[\bfx = x]}{\Pr[\bfy =
    y]} \geq \frac{\Pr[\bfx = x \wedge \bfy = y ]}{\Pr[\bfy = y]} =
  \Pr[\bfx = x \mid \bfy = y] > 2^{-k+s+t} .$$ %
  Therefore $\Pr[\bfy = y] < 2^{-s-t}$.  This implies that:
  $$\Pr[\bfy\text{ is bad}] = \sum_{y\text{ is bad}} \Pr[\bfy = y] < 2^s
  \cdot 2^{-s-t} = 2^{-t}$$
  which proves the lemma.
\end{proof}

Let $\bfx, \bfy$ be random variables over a common universe $X$.
Recall the following two equivalent definitions of statistical distance.
$$\Delta(\bfx, \bfy) = \frac{1}{2} \sum_{x \in X} |\Pr[\bfx = x] -
\Pr[\bfy = x]| = \max_{f:X \rightarrow \zo} |\Pr[f(\bfx) = 1] -
\Pr[f(\bfy) = 1]| .$$
Recall Pinsker's inequality:
\begin{lemma}
  \label{lem:pinsker}
  $\Delta(\bfx, \bfy) \leq \sqrt{\Div(\bfx \parallel \bfy) / 2}$ .
\end{lemma}

\subsection{Lower Bound on Communication Complexity of $\Line_p$}
To obtain a lower bound on $\PRDim(\Line_p)$ we prove that $\Rand_{1/5}^{\rightarrow,\pub}(\Eval_{\Line_p}) \geq \log p - O(1)$.

\begin{theorem}
  \label{thm:line-lb}
  $\Rand_{1/5}^{\rightarrow,\pub}({\Line_p}) \geq \log p - 7.$
\end{theorem}
\begin{proof}
  We set $\eps=1/5$ and let $\gamma = 2 (\frac{1}{2} - \frac{1}{p} - 2/5)^2$.  For $p \leq 128$ the claim obviously holds so we can assume that $p > 110$. This implies that $\gamma \geq 2/121$ and $\log (1/\gamma) \leq 6$.

  By the min-max principle, it suffices to exhibit an input
  distribution $\mu$ for which
  $$\Dist_\eps^{\rightarrow}(\Eval_{\Line_p}; \mu) \geq \log p -
  1 - \log \tfrac{1}{\gamma} \geq \log p - 7 .$$

  Let us consider the distribution $\mu$ that first samples $b \drawnr
  \zo$ and then outputs a sample from $\mu_b$ defined as follows.  We
  describe Alice's function $f$ by a pair $(a, b) \in \Z_p^2$ that
  defines a line.  The distribution $\mu_0$
  outputs uniform and independent pairs $(a,b), (x,y) \drawnr \Z_p^2$
  while $\mu_1$ outputs uniform $(a,b) \drawnr \Z_p^2$ and then uniform
  $(x,y)$ satisfying $ax + y = b$.

  It is clear that $\Pr_{(a,b,x,y) \sim \mu_0} [ax+y = b] = 1/p$
  while $\Pr_{(a,b,x,y) \sim \mu_1}[ax+y = b] = 1$.  Therefore, it
  must be for any protocol $\pi$ that computes $\Eval_{\Line_p}$ with
  overall error $\eps$ over $\mu$, the following must hold:
  \begin{align*}
    \frac{1}{2}\left(\Pr_{(a,b,x,y) \sim\mu_0}[\pi((a,b), (x,y)) = 1] - \frac{1}{p}\right) + \frac{1}{2}\left(1 -  \Pr_{(a,b,x,y) \sim\mu_1}[\pi((a,b), (x,y)) = 1]\right) \leq \eps.
  \end{align*}
  And therefore
  \begin{align}
     \Pr_{(a,b,x,y) \sim\mu_1}[\pi((a,b), (x,y)) = 1] - \Pr_{(a,b,x,y) \sim\mu_0}[\pi((a,b), (x,y)) = 1] \geq 1 - \frac{1}{p} - 2\eps.
\label{eq:gahpoidgh}
  \end{align}

  We will show that this is impossible for any one-way protocol $\pi$
  that communicates less than $c = \log p - 1 - \log
  \frac{1}{\gamma} \geq \log p - 7$ bits.  Fix any such $\pi$, and let $m$ be the message
  sent by Alice.

  Let us say that $m$ is \emph{good} if $\Ent_\infty(\bfa, \bfb \mid
  \bfm = m) \geq \log p + \log \frac{1}{\gamma} = 2\log p - c -
  1$. Here, all random variables are sampled according to
  $\mu_1$. Then by \lemmaref{lem:min-ent-compress} it holds that:
  \begin{align*}
    \Pr_{(a,b,x,y) \sim \mu_1, m = \pi(a,b)}[m \text{ is not good}] & = \Pr[\Ent_\infty(\bfa,
    \bfb \mid \bfm = m) < 2 \log p - c - 1]   \leq \tfrac{1}{2} .
  \end{align*}
  We next claim that:
  \begin{claim}
    \label{claim:collision}
    For any good $m$, $\Ent_2(\bfx, \bfy \mid \bfm = m) \geq 2 \log p
    - \gamma$ (with respect to the distribution $\mu_1$).
  \end{claim}
  We first prove the theorem using this claim.  It follows that
  $\Ent( \bfx, \bfy \mid \bfm = m) \geq \Ent_2( \bfx, \bfy \mid \bfm =
  m) \geq 2 \log p - \gamma$.  Observe that:
  $$\Info(\bfx, \bfy; \bfm \mid \bfm\text{ is good}) = \Ent(\bfx, \bfy \mid \bfm
  \text{ is good}) - \Ent( \bfx, \bfy \mid \bfm, \bfm \text{ is good}) \leq
  \gamma .$$
  By the divergence characterization of mutual information, this
  implies that
  \begin{equation}
    \Div(\bfx, \bfy, \bfm \mid \bfm \text{ is good} \parallel
    \angles{\bfx, \bfy} \angles{\bfm} \mid \bfm\text{ is good}) \leq \gamma .
    \label{eq:spodbinapoih}
  \end{equation}
  On the other hand, observe that the distribution $(\angles{\bfx,
    \bfy} \angles{\bfm} \mid \bfm\text{ is good})$ is identical to the
  distribution $(\bfx', \bfy', \bfm' \mid \bfm'\text{ is good})$ sampled
  according to the distribution $\mu_0$.  (The definition of $\bfm'$
  good is the same as for $\bfm$, since the marginal distribution of both
  $\mu_0$ and $\mu_1$ on the variables $\bfa, \bfb, \bfm$ is identical.)

  Therefore we have by Pinsker's inequality and
  \equationref{eq:spodbinapoih} that:
  $$\Delta\left( (\bfx, \bfy, \bfm \mid \bfm\text{ is good}),
    (\bfx', \bfy', \bfm' \mid \bfm'\text{ is good}) \right) \leq
  \sqrt{\gamma/2} < \frac{1}{2} - \frac{1}{p} - 2\eps .$$
  Thus, overall we have that:
  \begin{align*}
    \Delta\left((\bfx, \bfy, \bfm), (\bfx', \bfy', \bfm') \right) & \leq
    \Pr[\bfm \text{ is not good}]
    + \Delta\left((\bfx, \bfy, \bfm \mid \bfm\text{ is good}),
    (\bfx', \bfy', \bfm' \mid \bfm'\text{ is good}) \right) \\
  & < 1 - \frac{1}{p} - 2\eps .
  \end{align*}

  However, since the output of Bob depends only on $x, y, m$, it
  therefore follows that the probability that Bob outputs $1$ under
  $\mu_1$ is less than the probability that Bob outputs $1$ under
  $\mu_0$ plus $1 - \frac{1}{p} - 2\eps$, and this contradicts
  \equationref{eq:gahpoidgh}, proving the theorem.
\end{proof}

\begin{proof}[Proof of \claimref{claim:collision}]
  The joint distribution $\bfa, \bfb, \bfx, \bfy, \bfm$ drawn from
  $\mu_1$ can be viewed in the following order: first sample $\bfm$
  from the marginal distribution, then sample $\bfa, \bfb$
  conditioned on $\bfm$, and then sample $\bfx, \bfy$ a random point
  conditioned on $\bfa \bfx + \bfb = \bfy$.

  Conditioned on $\bfm = m$, consider $\bfx_1, \bfy_1$ and $\bfx_2,
  \bfy_2$ sampled independently as just described.  There are two ways
  a collision can occur: either $\bfa_1, \bfb_1$ and $\bfa_2, \bfb_2$
  collide or they do not.  In the first case $(\bfx_1, \bfy_1) =
  (\bfx_2, \bfy_2)$ occurs with probability $1/p$, in the second with
  probability at most $1/p^2$ since there is at most one point where
  the two lines intersect.

  We formalize this as follows:
  \begin{align*}
    \Pr[(\bfx_1, \bfy_1) = (\bfx_2, \bfy_2) \mid \bfm = m] & =
    \Pr[(\bfx_1, \bfy_1) = (\bfx_2, \bfy_2)  \wedge (\bfa_1, \bfb_1)
    = (\bfa_2, \bfb_2) \mid \bfm = m] \\
    & \quad +
    \Pr[(\bfx_1, \bfy_1) = (\bfx_2, \bfy_2) \wedge (\bfa_1, \bfb_1)
    \neq (\bfa_2, \bfb_2) \mid \bfm = m] \\
    & \leq \tfrac{1}{p} \Pr[(\bfa_1, \bfb_1) = (\bfa_2, \bfb_2) \mid
    \bfm = m] +  \tfrac{1}{p^2} \\
    & \leq \frac{1+\gamma}{p^2} .
  \end{align*}
  In the above we used the fact that $\Pr[(\bfa_1, \bfb_1) =
  (\bfa_2, \bfb_2) \mid \bfm = m] \leq \frac{\gamma}{p}$ because
  $m$ is good.

  Finally, $\Ent_2(\bfx, \bfy \mid \bfm = m) \geq \log
  \frac{p^2}{1+\gamma} > 2 \log p - \gamma$.
\end{proof}

\section{Separating pure and $(\alpha, \beta)$-differential privacy}
\label{sec:impure-summary}
We prove that it is possible to learn $\Line_p$ with $(\alpha,
\beta)$-differential privacy and $(\eps, \delta)$ accuracy using
$O(\frac{1}{\eps \alpha} \log \frac{1}{\beta} \log
\frac{1}{\delta})$ samples. This gives further evidence that it is possible to
obtain much better sample complexity with $(\alpha,
\beta)$-differential privacy than pure differential privacy.  Our
separation is somewhat stronger than that implied by our lower bound
for $\Thr_b$ and the upper bound of $O(16^{\log^*(b)})$ in
\citep{BeimelNS:13approx} since for $\Line_p$ we are able to
\emph{match} the non-private sample complexity (when the privacy and
accuracy parameters are constant\footnote{Formally a bound for
  constant $\beta$ is uninformative since weak $1/\beta$ dependence is
  achievable by naive subsampling. In our case the dependence on
  $1/\beta$ is logarithmic and we can ignore this issue.}), even
though, as mentioned in the previous section, randomized one-way
communication complexity and therefore the $\SCDP$ of $\Line_p$ is
asymptotically $\Theta(\log p)$. We note that our learner is not
proper since in addition to lines it may output point functions and
the all zero function.

\begin{theorem}
  \label{thm:improved}
  For any prime $p$, any $\eps, \delta, \alpha, \beta \in (0, 1/2)$,
  one can $(\eps, \delta)$-accurately learn $\Line_p$ with $(\alpha,
  \beta)$-differential privacy using $O(\frac{1}{\eps \alpha} \log
  \frac{1}{\beta} \log \frac{1}{\delta})$ samples.
\end{theorem}

We prove this theorem in two steps: first we construct a learner with
poor dependence on $\delta$ and then amplify using the exponential
mechanism to obtain a learner with good dependence on $\delta$.

\subsection{A learner with poor dependence on $\delta$}

\begin{lemma}
  \label{lem:impure}
  For any prime $p$, any $\eps, \delta, \alpha, \beta \in (0, 1/2)$,
  it suffices to take $O(\frac{1}{\eps} 2^{6/\delta} \cdot
  \frac{1}{\alpha} \log \frac{1}{\beta \delta})$ samples in order to
  $(\eps, \delta)$-learn $\Line_p$ with $(\alpha, \beta)$-differential
  privacy.
\end{lemma}
\begin{proof}
  At a high level, we run the basic (non-private) learner based on
  VC-dimension $O(\frac{1}{\alpha} \log \frac{1}{\beta})$ times.  We
  use the fact that $\Line_p$ is \emph{stable} in that after a
  constant number of samples, with high probability there is a
  \emph{unique} hypothesis that classifies the samples correctly.
  (This is simply because any two distinct points on a line define the
  line.)  Therefore, in each of the executions of the non-private
  learner, we are likely to recover the same hypothesis.  We can then
  release this hypothesis $(\alpha, \beta)$-privately using the
  ``Propose-Test-Release'' framework.

  The main challenge in implementing this intuition is to eliminate
  corner cases, where with roughly probability $1/2$ the sample set
  may contain two distinct positively labeled points and with
  probability $1/2$ only a single positively labeled point, as this
  would lead to unstable outputs.  We do this by randomizing the
  \emph{number} of samples we take.

  Let $t$ be a number of samples, to be chosen later.  Given $t$
  samples $(x_1, y_1), \ldots, (x_t, y_t)$, our \emph{basic learner}
  will do the following:
  \begin{enumerate}
  \item See if there exist two distinct samples $(x_i, y_i) \neq (x_j,
    y_j)$ that are both classified positively.  If so, output the
    unique line defined by these points.
  \item Otherwise, see if there exists any sample $(x_i, y_i)$
    classified positively.  Output the point function that outputs $1$
    on $(x_i, y_i)$ and zero elsewhere.
  \item Otherwise, output the constant $0$ hypothesis.
  \end{enumerate}

  \newcommand{\freq}{\mathsf{freq}} %
  Our \emph{overall learner} uses the basic learner as follows: first
  sample an integer $k$ uniformly from the interval
  $[\log(\ln(3/2)/\eps), \log(\ln(3/2)/\eps) + 6/\delta]$ and set
  $t = 2^k$. Set $\ell = \max\{ \frac{12}{\alpha} \ln \frac{2}{\beta
    \delta} + 13, 72 \ln \frac{4}{\delta} \}$.  Set $n = t \ell$.
  \begin{enumerate}
  \item Take $n$ samples and cut them into $\ell$ subsamples of size
    $t$, and run the basic learner on each of these.
  \item Let the returned hypotheses be $h_1, \ldots, h_\ell$.  Define
    $\freq(h_1, \ldots, h_\ell) = \argmax_h |\{h_i = h \mid i \in
    [\ell]\}|$, \ie the most frequently occurring hypothesis, breaking
    ties using lexicographical order.  We define $\overline{h} = \freq(h_1,
    \ldots, h_\ell)$.  Compute $c$ to be the smallest number of $h_i$
    that must be changed in order to change the most frequently
    occuring hypothesis, \ie
    $$c = \min \left\{ c \mid \exists h'_1, \ldots, h'_\ell, \freq(h'_1,
      \ldots, h'_\ell) \neq \overline{h}, c = |\{i \mid h_i \neq
      h'_i\}| \right\} .$$
  \item If $c + \Lambda(1/\alpha) > \frac{1}{\alpha} \ln \frac{1}{2\beta}
    + 1$ then output $\overline{h}$, otherwise output the constant $0$
    hypothesis.
  \end{enumerate}
  Here, $\Lambda(1/\alpha)$ denotes the Laplace distribution, whose
  density function at point $x$ equals $\alpha e^{- \alpha |x|}$.  It
  is easy to check that adding $\Lambda(1/\alpha)$ to a sum of Boolean
  values renders that sum $\alpha$-differentially private
  \citep{DworkMNS:06}.

  We analyze the overall learner.  Observe that once $t$ is fixed, the
  basic learner is deterministic.

  \paragraph{Privacy:} we prove that the overall learner is $(\alpha,
  \beta)$-differentially private.  Consider any two neighboring
  inputs $x, x' \in (\Z_p^2 \times \zo)^n$.  There are two cases:
  \begin{itemize}
  \item The most frequent hypothesis $\overline{h}$ returned by
    running the basic learner on the $\ell$ subsamples of $x, x'$ is
    the same.  In this case, there are two possible outputs of the
    mechanism, either $\overline{h}$ or the $0$ hypothesis.  Due to
    the fact that we decide between them using a count with Laplace
    noise and the count has sensitivity $1$, the probability assigned
    to either output changes by at most a multiplicative $e^{-\alpha}$
    factor between $x, x'$.
  \item The most frequent hypotheses are different.  In this case $c =
    1$ for both $x, x'$.  The probability of \emph{not} outputting
    $0$ in either case is given by
    $$\Pr[\Lambda(1/\alpha) > \frac{1}{\alpha} \ln \frac{1}{2\beta}] = \beta .$$
    Otherwise, in both cases they output $0$.
  \end{itemize}

  \paragraph{Accuracy:} we now show that the overall learner
  $(\eps,\delta)$-PAC learns.  We claim that:
  \begin{claim}
    \label{claim:randomt}
    Fix any hidden line $f$ and any input distribution $\calD$.  With
    probability $1-\delta/2$ over the choice of $t$, there is a
    \emph{unique} hypothesis with error $\leq \eps$ that the basic
    learner will output with probability at least $2/3$ when given $t$
    independent samples from $\calD$.
  \end{claim}
  Let us first assume this claim is true.  Then it is easy to show
  that the overall learner $(\eps, \delta)$-learns: suppose we are in
  the $1-\delta/2$ probability case where there is a unique hypothesis
  with error $\leq \eps$ output by the basic learner.  Then, by
  Chernoff, since $\ell \geq 72 \ln \frac{4}{\delta}$ it holds that
  with probability $1-\delta/4$ at least $7/12$ fraction of the basic
  learner outputs will be this unique hypothesis.  This means that the
  number of samples that must be modified to change the most frequent
  hypothesis is $c \geq \frac{\ell}{12}$.  Therefore since $\ell \geq
  \frac{12}{\alpha} \ln \frac{1}{\beta \delta} + 13$, in this case the
  probability that the overall learner does not output this unique
  hypothesis is bounded by:
  $$\Pr[c + \Lambda(\tfrac{1}{\alpha}) \leq \tfrac{1}{\alpha} \ln
  \tfrac{1}{2\beta} + 1] \leq
  \Pr[\Lambda(\tfrac{1}{\alpha}) < - \tfrac{1}{\alpha} \ln
  \tfrac{2}{\delta}] = \tfrac{\delta}{4} .$$
  Thus the overall probability of not returning an $\eps$-good
  hypothesis is at most $\delta$.

  \newcommand{\None}{\mathsf{None}}
  \newcommand{\One}{\mathsf{One}}
  \newcommand{\Two}{\mathsf{Two}}

  \paragraph{Proof of \claimref{claim:randomt}}  Fix a concept $f$ defined by
  a line given by $(a, b) \in \Z_p^2$ and any input distribution
  $\calD$ over $\Z_p^2$.

  Define the following events $\None_t, \One_t, \Two_t$ parameterized
  by an integer $t > 0$ and defined over the probability space of
  drawing $(x_1, y_1), \ldots, (x_t, y_t)$ independently from $\calD$:
  \begin{itemize}
  \item $\None_t$ is the event that all of the $(x_i, y_i)$ are not on
    the line $(a,b)$.
  \item $\One_t$ is the event that there exists some
    $(x_i, y_i)$ on the line $(a,b)$, and furthermore for every other
    $(x_j, y_j)$ on the line $(a,b)$ is in fact equal to $(x_i, y_i)$.
  \item $\Two_t$ is the event that there exists distinct $(x_i, y_i)
    \neq (x_j, y_j)$ that are both on the line $(a,b)$.
  \end{itemize}
  Next we will show that with probability $1-\delta/2$ over the choice
  of $t$, one of these three events has probability at least $2/3$,
  and then we show that this suffices to imply the claim.

  Let $r = \Pr_{(x,y) \sim \calD}[f(x,y) = 1]$, let $q_{x,y} =
  \Pr_{(x',y') \sim \calD}[(x',y') = (x, y)]$, and let $q =
  \max_{(x,y) \in f^{-1}(1)} q_{x,y}$.  We can characterize the
  probabilities of $\None_t, \One_t, \Two_t$ in terms of $r, q, t$ as
  follows:
  \begin{align*}
    \Pr[\None_t] & = (1 - r)^t , \\
    \Pr[\One_t] & = \sum_{(x,y) \in f^{-1}(1)} ((1-r+q_{x,y})^t - (1-r)^t) , \\
    \Pr[\Two_t] & = 1 - \Pr[\None_t] - \Pr[\One_t] .
  \end{align*}
  The characterizations for $\None_t, \Two_t$ are obvious.  The
  characterization of $\One_t$ is exactly the probability over all
  $(x,y) \in f^{-1}(1)$ that all samples are either labeled $0$ or
  equal $(x,y)$, excluding the event that they are all labeled $0$.

  From the above and by considering the $(x,y)$ maximizing $q_{x,y}$,
  we have the following bounds:
  \begin{align}
    \Pr[\None_t] & \geq 1 - rt , \label{eq:none} \\
    \Pr[\One_t] & \geq (1-r+q)^t - (1-r)^t \geq 1 - (r-q)t - e^{-rt}
    \label{eq:one} ,\\
    \Pr[\Two_t] & \geq (1 - e^{-rt/2})(1 - e^{-(r-q)t/2}) .\label{eq:two}
  \end{align}
  The first two follow directly from the fact that for all $x \in \R$
  it holds that $1-x \leq e^x$ and also for all $x \in [0, 1]$ and $y
  \geq 1$ it holds that $(1-x)^y \geq 1 - x y$.  \equationref{eq:two}
  follows from the following argument.  $\Two_t$ contains the
  sub-event where there is at least one positive example in the first
  $t/2$ samples and a different positive example in the second $t/2$
  samples.  The probability of this sub-event is lower-bounded
  by $(1 - (1-r)^{t/2})(1 - (1-r+q)^{t/2}) \geq (1 - e^{-rt/2})(1 -
  e^{-(r-q)t/2})$.


  \paragraph{$t$ is good with high probability.} Let us say that $t$
  is good for $\None_t$ if $t \leq \frac{1}{3r}$.  We say $t$ is good
  for $\One_t$ if $t \in [\frac{\ln 6}{r}, \frac{1}{6(r-q)}]$.  We say
  $t$ is good for $\Two_t$ if $t \geq \frac{2 \ln 6}{r-q}$.  (It is
  possible that some of these events may be empty, but this does not
  affect our argument.)  Using \equationref{eq:none}, \equationref{eq:one} and
  \equationref{eq:two}, it is clear that if $t$ is good for some event,
  then the probability of that event is at least $2/3$.

  Let us say $t$ is good if it is good for any one of $\None_t,
  \One_t, \Two_t$. $t$ is good means the following when viewed on the
  logarithmic scale:
  $$\log t \in [0, \log \tfrac{1}{r} - \log 3] \ \cup\ [\log
  \tfrac{1}{r} + \log \ln 6, \log \tfrac{1}{r-q} - \log 6] \ \cup\ [\log \tfrac{1}{r-q} + \log (2
  \ln 6), \infty) .$$
  But this means that $t$ is bad on the logarithmic scale is equivalent to:
  \begin{equation}
    \label{eq:logtbad}
    \log t \in (\log \tfrac{1}{r} - \log 3, \log \tfrac{1}{r} + \log \ln
    6) \ \cup \ (\log \tfrac{1}{r-q} - \log 6, \log \tfrac{1}{r-q} +
    \log (2 \ln 6)) .
  \end{equation}
  Thus, for any $r$, there are at most $3$ integer values of $\log t$
  that are bad.  But recall that $t = 2^k$ where $k$ is uniformly
  chosen from $\{\log(\ln(3/2)/\eps), \ldots, \log(\ln(3/2)/\eps)
  + 6/\delta\}$.  Therefore the probability that $k = \log t$ is one of
  the bad values defined in \equationref{eq:logtbad} is at most $\delta/2$.

  \paragraph{When $t$ is good, basic learner outputs unique accurate
    hypothesis.} To conclude, we argue that when $t$ is good then the
  basic learner will output a unique hypothesis with error $\leq
  \eps$ with probability $\geq 2/3$.  This is obvious when $t$ is
  good for $\Two_t$, since whenever the basic learner sees two points
  on the line, it recovers the exact line.  It is also easy to see
  that when $t$ is good for $\None_t$, the basic learner outputs the
  $0$ hypothesis with probability $2/3$, and this has error at most
  $\eps$ since
  $$2/3 \leq \Pr[\None_t] = (1-r)^t \leq e^{-rt} \Rightarrow r \leq
  \ln(3/2) / t \leq \eps .$$

  It remains to argue that the basic learner outputs a unique
  hypothesis with error at most $\eps$ when $t$ is good for
  $\One_t$.  Observe that we have actually set the parameters so that
  when $t$ is good for $\One_t$, it holds that:
  \begin{equation}
    \label{eq:oneactual}
    \Pr[\One_t \wedge \text{unique positive point is }(x_{\max}, y_{\max})]
    \geq 2/3 ,
  \end{equation}
  where $(x_{\max}, y_{\max}) = \argmax_{(x,y) \in f^{-1}(1)}
  q_{x,y}$.  Therefore, for such $t$, the basic learner will output
  the point function that is positive on exactly $(x_{\max},
  y_{\max})$ with probability at least $2/3$.

  To show that this point function has error at most $\eps$, it
  suffices to prove that
  $$\Pr[f(x,y) = 1 \wedge (x,y) \neq (x_{\max},
  y_{\max})] = r - q \leq \eps .$$ %
  From \equationref{eq:oneactual}, we deduce that:
  $$2/3 \leq (1 - r + q)^t - (1-r)^t \leq e^{-(r-q)t} \Rightarrow r-q
  \leq \ln(3/2) / t \leq \eps .$$
  This concludes the proof.
\end{proof}

\noindent
{\bf Improving dependence on $\delta$:}
We now improve the exponential dependence on $1/\delta$ in
\lemmaref{lem:impure} to prove \theoremref{thm:improved}.
We will use the algorithm of \lemmaref{lem:impure} with $\delta = 1/2$
and accuracy $\eps/2$ repeated $k=O(\log(1/\delta))$ times independently in order to construct a set $H$ of $k$ hypotheses. We then draws a fresh sample $S$ of $O(\log(1/\delta)/(\eps\al))$ examples and select one of the hypotheses based on their error on $S$ using the exponential mechanism of \citep{McSherryTalwar:07}. This mechanism chooses a hypothesis from $H$ with probability proportional to $e^{ -\alpha\cdot\err_S(h) / 2}$, where $\err_S(h)$ is $\err_S(h) = |\{ (x,\ell) \in S \mid h(x) \neq \ell \}|$.
Simple analysis \citep[\eg][]{KLNRS:11,BeimelNS:13} then shows that the selection mechanism is $\alpha$-differentially private and outputs a hypothesis that has error of at most $\eps$ on $\calD$ with probability at least $1-\delta$.
Note that each of the $k$ copies of the low-confidence algorithm and
the exponential mechanism are run on disjoint sample sets and
therefore there is no privacy loss from such composition. Hence the
resulting algorithm is also $(\alpha,\beta)$-differentially
private. We include formal details in \sectionref{app:boost-delta}.

\section{Conclusions and Open Problems}
Our work continues the investigation of the costs of privacy in standard classification models initiated by \citet{KLNRS:11}. Our main result is a new connection to communication complexity that provides a rich set of results and techniques developed in the past 30 years in communication complexity to the study of differentially private learning. Most notably, we show that tools from information theory can be used to resolve several fundamental questions about $\SCDP$. Our connection relies on the characterization of $\SCDP$ using natural notions of representation dimension introduced by \citet{BeimelNS:13}. We remark, however, that our lower bounds can also be proved more directly without relying on the results in \citep{BeimelNS:13}. Implicit in our lower bounds is a lower bound on the mutual information between random examples and the hypothesis output by the private learning algorithm. On the other hand, it is known and easy to show that $\alpha$-differential privacy gives an upper bound of  $O(\alpha \cdot n)$ on the mutual information between $n$ data points and the output of the algorithm (see for example \citep{DworkFHPRR15:arxiv}).

While we focus on differential privacy our lower bounds have immediate implications in other settings where information about the examples needs to be stored or transmitted for the purposes of classification, such as distributed computation, streaming and low-memory computation.

Our work also demonstrates that PAC learning with approximate differential privacy can be substantially more sample efficient then learning with pure differential privacy. However our understanding of classification with approximate differential privacy still has some major gaps. Most glaringly, we do not know whether the sample complexity of $(\alpha,\beta)$-differentially private PAC learning is different from the VC dimension (up to a $\poly(1/\alpha,\log(1/\beta))$ factor).  Also the separation from the pure differential privacy holds only for PAC learning and we do not know if it is also true for agnostic learning.

\section*{Acknowledgements}
We are grateful to Kobbi Nissim for first drawing our attention to the intriguing problem of
understanding the relationship between probabilistic representation
dimension and VC dimension, and for valuable discussions regarding the
sample complexity of privately learning threshold functions.
We thank Nina Balcan and Avrim Blum who brought up the relationship of
our bounds for intervals in \sectionref{sec:cc-equiv-app} to those based
on Littlestone's dimension. Their insightful comments and questions
have lead to our result in \theoremref{thm:ldinformal}. We also thank Sasha Rakhlin and Sasha Sherstov for useful suggestions and references.

D.X. was supported in part by the French ANR Blanc program under contract
ANR-12-BS02-005 (RDAM project), by NSF grant CNS-1237235, a gift from
Google, Inc., and a Simons Investigator grant to Salil Vadhan.

\bibliographystyle{abbrvnat}
\bibliography{ppac-cc}


\appendix

\section{Proof of \lemmaref{lem:lb-thr}}
\label{sec:ppac-cc-ai}

As before, for $\eps \in [0, 1]$, we let $\Ent(\eps)$ denote the binary
entropy of $\eps$, namely the entropy of the Bernoulli random variable
that equals $1$ with probability $\eps$.  We recall Fano's inequality
(for the Boolean case):
\begin{lemma}[Fano's inequality]
  Let $\bfx, \bfy$ be Boolean random variables such that $\Pr[ \bfx =
  \bfy ] \geq 1 - \eps$.  Then it holds that $\Ent(\bfx \mid \bfy)
  \leq \Ent(\eps)$.
\end{lemma}
We can now prove \lemmaref{lem:lb-thr}.
\begin{proof}
 By the Yao's \citeyearpar{Yao:1977} min-max principle ,
  $$\Rand^{\rightarrow}_\eps(\AI) \geq \Dist^{\rightarrow}_\eps(\AI;
  \mu) ,$$
  where $\mu$ is the input distribution that samples uniform $\bfx
  \drawnr \zo^d$ and uniform $\bfi \drawnr [d]$.

  Consider any deterministic protocol $\pi$ computing $\AI$ with error
  at most $\eps$ over $\mu$, and suppose that $\pi$ uses $\delta \cdot
  d$ communication.  Let $(\bfx, \bfi) \sim \mu$.  Then
  $\Info(\pi_A(\bfx) ; \bfx) \leq \delta \cdot d$.  By the chain rule for mutual information it follows that
  \begin{align*}
    \delta \cdot d & \geq \Info(\pi_A(\bfx); \bfx) \\
    & = \sum_{i=1}^d \Info(\bfx_i; \pi_A(\bfx) \mid \bfx_1, \ldots, \bfx_{i-1}) \\
    & = \sum_{i=1}^d (\Ent(\bfx_i \mid \bfx_1, \ldots, \bfx_{i-1}) - \Ent(\bfx_i \mid \pi_A(\bfx), \bfx_1, \ldots, \bfx_{i-1})) \\
    & = \sum_{i=1}^d (1 - \Ent(\bfx_i \mid \pi_A(\bfx), \bfx_1, \ldots, \bfx_{i-1})) .
  \end{align*}
  We therefore deduce that (for $\bfi$ uniform over $[d]$):
  \begin{equation}
    \label{eq:aspgdoiha}
    1 - \Ent(\bfx_\bfi \mid \pi_A(\bfx), \bfx_1, \ldots, \bfx_{\bfi-1}, \bfi) \leq
    \delta .
  \end{equation}
  By Fano's Inequality, we know that if the probability of guessing
  $\bfx_\bfi$ given $\pi_A(\bfx), \bfx_1, \ldots, \bfx_{\bfi-1}, \bfi$ (which is exactly Bob's
  input) is at least $1 - \eps$, then $\Ent(\bfx_\bfi \mid \pi_A(\bfx), \bfx_1,
  \ldots, \bfx_{\bfi-1}, \bfi) \leq \Ent(\eps)$.  From this and
  \autoref{eq:aspgdoiha}, we deduce that $\delta \geq 1 - \Ent(\eps)$.
\end{proof} 

\section{Ldim Lower Bound for Halfspaces}
\label{sec:ldim-hs}

Recall that $\HS_b^d$ denotes the concept class of all halfspaces over $I_b^d$, where $I_b=\{0,1,\ldots,2^b-1\}$. Our proof is based on the technique used in \citep{Muroga:71} to prove a lower bound of $2^{d(d-1)}$ on the total number of distinct halfspaces over $\zo^d$. As a first step we prove the following simple lemma (for $b=1$ it can also be found in \citep{Muroga:71}).
\begin{lemma}
\label{lem:hs2order}
For an integer $b\geq 1$ let $f$ be a halfspace over $I_b^d$. There exists a vector $w \in \Z^d$ and an integer $\theta \in \Z$ such that:
\begin{itemize}
\item $(w,\theta)$ represents $f$, that is $f(x)=1$ if and only if $ w \cdot x \geq \theta$;
\item for every two distinct $x,x' \in I_b^d$, $w \cdot x \neq w \cdot x'$.
\end{itemize}
We refer to such a representation of $f$ as {\em collision-free}.
\end{lemma}
\begin{proof}
Let $(w',\theta')$ be any integer weight representation of $f$ (such representation always exists for a halfspace over integer points). We first create a margin around the decision boundary by setting $w'' = 2^{d+1} w'$ and $\theta = 2^{d+1} \theta' - 2^d$. Note that $(w'',\theta)$ also represents $f$ and, in addition, for every $x$, $|w'' \cdot x - \theta| \geq 2^d$. We now define for every $i\in [d]$, $w_i = w''_i + 2^i$.  It is easy to see that $|w \cdot x - w'' \cdot x| \leq 2^d-1$ and therefore $(w,\theta)$ represents $f$. Further, $d$ least significant bits in the binary representation of $w \cdot x$ are exactly equal to $x$ and therefore the second condition is also satisfied.
\end{proof}

Let $(w,\theta)$ be some fixed collision-free representation of a halfspace $f$. By ordering the elements of $I_b^d$ according to the value of $w \cdot x$ we obtain a strict order over the $(2^b)^d$ elements of $I_b^d$. Further any threshold function on this order is of the form $w \cdot x \geq \theta'$ for some $\theta' \in \Z$. We exploit this observation to embed threshold functions into halfspaces. We can then use the well-known fact that $\Ldim$ of threshold functions over an interval of size $2^b$ is $b$.
\begin{theorem}
$\Ldim(\HS_b^d) = (d(d-1)/2+1)\cdot b$.
\end{theorem}
\begin{proof}
We construct a mistake tree $T_d$ over $\HS_b^d$ and $I_b^d$ inductively over the dimension $d$. For $d=1$, $\HS_b^d$ includes all threshold functions on $I_b$ and therefore we define $T_1$ is the complete binary tree representing the binary search on this interval. Note that the depth of this tree is $b$.

Now for $d \geq 2$, let $T_{d-1}$ be the complete mistake tree over $\HS_b^{d-1}$ and $I_b^{d-1}$ of depth $((d-1)(d-2)/2+1)\cdot b$ given by our inductive construction. For every leaf $\ell$ let $f_\ell \in \HS_b^{d-1}$ be the halfspace labeling the leaf. Let $(w',\theta)$ be a collision-free representation  of $f_\ell$ (arbitrarily chosen but fixed for every possible halfspace). Let $Z_\ell = \{ \  w' \cdot y \ | \ y \in I_b^{d-1}\}$. The collision-free property of $(w',\theta)$ implies that $|Z_\ell| = |I_b|^{d-1} = 2^{b (d-1)}$. Let $z_0<z_1\ldots <z_{|Z_\ell|-1}$ denote the elements of $Z_\ell$ ordered by value and for every $j \leq 2^{b (d-1)}$, let $y^j$ denote the point $y$ such that $w' \cdot y = z_j$. For every $z \in Z_\ell$ let $f_{\ell,z}$ be the halfspace over $I_b^d$ defined by $(w,\theta)$ where, $w_d = \theta - z$ and $w_i = w'_i$ for all $i \leq d-1$. Clearly,  $f_{\ell,z_j}$ restricted to the $d-1$ dimensional subcube $I_b^{d-1} \times \{0\}$ (that is points $x$ in $I_b^d$ for which $x_d = 0$) is equivalent to $f_\ell$. When restricted to the $d-1$ dimensional subcube $I_b^{d-1} \times \{1\}$, $f_{\ell,{z_j}}$ is equivalent to $w' \cdot y \geq z_j$. Therefore, up to renaming of the points $y^j \rightarrow j$ and functions $f_{\ell,z_j} \rightarrow t_j$,
$F_{\ell} = \{f_{\ell,z}\}_{z \in Z_\ell}$ restricted to $I_b^{d-1} \times \{1\}$ is identical to the class of linear thresholds on interval $I_{b(d-1)}=\{0,1,\ldots,2^{b(d-1)}\}$. This means that there exists a complete mistake tree $T_\ell$ for $F_{\ell}$ over $I_b^{d-1} \times \{1\}$ of depth $b(d-1)$.

Let $T_d$ be the mistake tree obtained by $(a)$ starting with $T_{d-1}$; $(b)$ replacing points in $I_b^{d-1}$ that label nodes by the corresponding points in $I_b^{d-1} \times \{0\}$; $(c)$ replacing each leaf $\ell$ of $T_{d-1}$ with $T_\ell$. We claim that this is a complete mistake tree for $\HS_b^{d}$ over $I_b^d$ of depth $(d(d-1)/2+1)\cdot b$. The fact that this tree is a complete binary tree of depth $(d(d-1)/2+1)\cdot b$ follows immediately from our construction since  $((d-1)(d-2)/2+1)\cdot b + (d-1)b = (d(d-1)/2+1)\cdot b$. Now let $\ell'$ be a leaf of $T_d$ labeled by some halfspace $f_{\ell,z_j}$. Let $v'$ be a node in $T_d$ labeled by point $x$ such that $\ell'$ is in the subtree of $v'$. If $v'$ is a node derived from node $v$ in $T_{d-1}$ then, by definition, $\ell$ is a leaf in the subtree of $v$ in $T_{d-1}$ and $v$ is labeled by $y$ such that $x = y0$. By our construction, $f_{\ell,z_j}(y0) = f_\ell(y)$ and therefore $f_{\ell,z_j}(x)=1$ if and only if $\ell$ is in the right subtree of $v$ which is equivalent to $\ell'$ being in the right subtree of $v'$.

If $v'$ is a node in $T_\ell$, then $x \in  I_b^{d-1} \times \{1\}$. On points in $I_b^{d-1} \times \{1\}$ the function $f_{\ell,z_j}$ corresponds to the threshold function $t_j$ on the interval $I_{b(d-1)}$ and
consistency with $T_d$ follows from the properties of the binary search tree $T_\ell$ for threshold functions.
\end{proof}

\section{Improving Dependence on $\delta$ in \theoremref{thm:improved}}
\label{app:boost-delta}
\newcommand{\EM}{\mathsf{EM}} %
We now improve the exponential dependence on $1/\delta$ in
\autoref{lem:impure} to prove \autoref{thm:improved}.  We first
introduce the exponential mechanism of \cite{McSherryTalwar:07}.
 For simplicity we only describe its
restriction to the learning setting.  Let $H$ be a hypothesis class
and define the ``quality score function'' $q(S, h) = |\{ (x,y) \in S
\mid h(x) = y \}|$.  The exponential mechanism for $q$ with privacy
$\alpha$ is the following: given an input $S \in (X \times \zo)^n$,
output $h \in H$ according to the distribution $\EM(S)$ given by
$$\Pr[\EM(S) = h] \propto e^{\alpha q(S, h) / 2} .$$
We use the following theorem about the exponential mechanism.
Let $q_{\max}(S) = \max_{h \in H} q(S, h)$.
\begin{theorem}[\citealp{McSherryTalwar:07}]
  \label{thm:em}
  The exponential mechanism is $\alpha$-differentially private.
  Furthermore, for all $S \in (X \times \zo)^n$ and all $t > 0$, it holds that
  $$\Pr[q(S, \EM(S)) < q_{\max}(S) - t] \leq |H| e^{-\alpha t / 2} .$$
\end{theorem}

We now finish the proof of \autoref{thm:improved}.
\begin{proof}[Proof of \autoref{thm:improved}]
  We will use the algorithm of \autoref{lem:impure} with $\delta =
  \Theta(1)$ in order to construct a small set of hypotheses from
  which we'll then select one using the exponential mechanism.  More
  precisely:
  \begin{enumerate}
  \item Set $k = \log(2/\delta) / \log(4/3)$. Run the algorithm of
    \autoref{lem:impure} $k$ times independently with fresh samples
    and with $(\frac{\eps}{4}, \frac{1}{4})$-accuracy and $(\alpha,
    \beta)$-differential privacy.  Call the resulting hypotheses $H =
    \{h_1, \ldots, h_k\}$.
  \item Sample $m = \frac{16}{\eps\al} \log \frac{4k}{\delta}$
    additional samples, call this set $S$.  Use the exponential
    mechanism to output $h \in H$.
  \end{enumerate}

  The mechanism is $(\alpha, \beta)$-differentially private because
  samples used to produce $h_1, \ldots, h_k$ are learned with
  $(\alpha, \beta)$ differential privacy (notice each sample can only
  affect one of the $h_i$, therefore the privacy loss does not add up
  when considering the set of all hypotheses).  Also, the samples used
  to pick $h \in H$ are used via the exponential mechanism, which is
  also $\alpha$-differentially private.

  To analyze the accuracy, observe that since each of the $h_i$ is
  produced using an $(\frac{\eps}{4}, \frac{1}{4})$-accurate learner,
  therefore by independence of the executions and our choice of $k$,
  it holds with probability $\geq 1-(1-\frac{1}{4})^k \geq 1-
  \frac{\delta}{2}$ that $H$ contains some $h$ that has error at most
  $\frac{\eps}{4}$.

  Next, for any $h \in H$ with $\Pr_{x \sim \calD}[f(x) \neq h(x)] >
  \eps$, observe that by a standard multiplicative Chernoff bound with
  probability $1 - e^{-\eps m/12}$ over the choice of $S$ it holds
  that
  \begin{equation}
    \label{eq:anposbiha}
    \Pr_{x \drawnr S}[f(x) \neq h(x)] > \tfrac{1}{2} \Pr_{x \sim
    \calD}[f(x) \neq h(x)] > \eps / 2 .
  \end{equation}
  Similarly, for any $h \in H$ with $\Pr_{x \sim \calD}[f(x) \neq
  h(x)] \leq \eps / 4$, with probability $1 - e^{-\eps m / 8}$ over the
  choice of $S$ it holds that:
  \begin{equation}
    \label{eq:anposbiha2}
    \Pr_{x \drawnr S}[f(x) \neq h(x)] < \tfrac{3}{2} \Pr_{x \sim
    \calD}[f(x) \neq h(x)] \leq 3\eps/8 .
  \end{equation}
  Therefore by a union bound, it holds with probability $1 - k
  e^{-\eps m / 12} > 1 - \frac{\delta}{4}$ that for all $h \in H$ with
  error greater than $\eps$, it holds that $\Pr_{x \drawnr S}[f(x)
  \neq h(x)] > \eps / 2$, and for all $h \in H$ with error at most
  $\eps/4$, it holds that $\Pr_{x \drawnr S}[f(x) \neq h(x)] < 3\eps/8$.

  This implies that with probability $1 - 3\delta/4$ over the
  probability of computing $H$ and sampling $S$, it suffices to output
  some $h \in H$ such that $q(S, h) \geq \max_{h' \in H} q(S, h') -
  \eps/8$.  This is because we are in the case where $H$ contains a
  hypothesis with error $\leq \eps/4$, and therefore by
  \autoref{eq:anposbiha2} it holds that $\max_{h' \in H} q(S, h') >
  |S| ( 1 - 3\eps / 8)$ and therefore any such $h$ will have error
  $\leq \eps/2$ over $S$.  By \autoref{eq:anposbiha}, we deduce that
  any such $h$ must have error $\leq \eps$ over $\calD$.

  By \autoref{thm:em}, the probability that the exponential mechanism
  outputs such a $h$ is at least $1 - k e^{-\alpha \eps m/16} \geq 1 -
  \delta / 4$.  Therefore the overall probability of outputting an
  $\eps$-good hypothesis is at least $1 - \delta$.
\end{proof}

\end{document}